%% file: main.tex
\documentclass[11pt, letterpaper]{article}
\usepackage[utf8]{inputenc}
\usepackage{fullpage,times}
\usepackage[T1]{fontenc}

\usepackage{caption}
\usepackage{algorithm}
\usepackage{algorithmic}
\renewcommand{\thealgorithm}{}
\usepackage{thmtools}
\usepackage{thm-restate}
\usepackage{amsmath,amsfonts,amsthm,amssymb,multirow}
\usepackage{mathtools}
\usepackage{times}
\usepackage{graphicx}
\usepackage{floatpag}
\usepackage{bbold}
\usepackage{enumitem}

\usepackage{calc}
\usepackage{tikz}
\usepackage{hyperref}

\usepackage{graphicx}
\usepackage{tabularx}
\usepackage[skins,breakable]{tcolorbox}

\newtheorem{theorem}{Theorem}[section]

\newtheorem{observation}{Observation}[section]
\newtheorem{claim}{Claim}[section]

\usepackage[capitalize,nameinlink]{cleveref}
\crefname{theorem}{Theorem}{Theorems}
\Crefname{lemma}{Lemma}{Lemmas}
\Crefname{claim}{Claim}{Claims}
\Crefname{observation}{Observation}{Observations}
\Crefname{algorithm}{Algorithm}{Algorithms}
\Crefname{myalgctr}{Algorithm}{Algorithms}
\Crefname{challenge}{Challenge}{Challenges}
\Crefname{figure}{Figure}{Figures}
\Crefname{hypothesis}{Hypothesis}{Hypotheses}
\theoremstyle{definition}
\newtheorem{definition}{Definition}[section]
\newtheorem{notation}{Notation}[section]

\DeclareMathOperator\dist{dist}
\DeclareMathOperator\len{len}
\newcommand\eps{\varepsilon}
\newcommand{\otil}{\tilde{O}}
\newcommand{\gstar}{G^{*}}

\makeatletter
\newenvironment{breakablealgorithm}
  {% \begin{breakablealgorithm}
   \begin{center}
     \refstepcounter{algorithm}% New algorithm
     \hrule height.8pt depth0pt \kern2pt% \@fs@pre for \@fs@ruled
     \renewcommand{\caption}[2][\relax]{% Make a new \caption
       {\raggedright\textbf{\fname@algorithm~\thealgorithm} ##2\par}%
       \ifx\relax##1\relax % #1 is \relax
         \addcontentsline{loa}{algorithm}{\protect\numberline{\thealgorithm}##2}%
       \else % #1 is not \relax
         \addcontentsline{loa}{algorithm}{\protect\numberline{\thealgorithm}##1}%
       \fi
       \kern2pt\hrule\kern2pt
     }
  }{% \end{breakablealgorithm}
     \kern2pt\hrule\relax% \@fs@post for \@fs@ruled
   \end{center}
  }
\makeatother

\title{Closing the Gap Between Directed Hopsets and Shortcut Sets}

\author{Aaron Bernstein\thanks{Rutgers University. Funded by NSF CAREER grant 1942010.} \and Nicole Wein\thanks{DIMACS, Rutgers University. Supported by a grant to DIMACS from the Simons Foundation (820931).}}

\date{}
\begin{document}

\maketitle

\pagenumbering{gobble}

\begin{abstract}
\input{abstract.tex}

\end{abstract}

\clearpage
\tableofcontents
\clearpage

\pagenumbering{arabic}

%\maketitle
\input{intro.tex}

%\begin{lemma}
%\label{assumption-weight}
%Trhoughout the paper, we assume that in the input graph $G$ all edge weights are in the range $[1,n^2/\epsilon^2]$, where $(1+\epsilon)$ is the desired final approximation bound.
%\end{lemma}

%\paragraph{Justification of \cref{assumption-weight}}
%It is not hard to reduce the case of an arbitrary graph to the case of a graph that satisfies \cref{assumption-weight}. Let $G$ be a graph that does not necessarily satisfy this assumption and let $W$ be the maximum edge weight in $G$. 

\section{Technical Overview}
\label{sec_overview}

Recall that \cite{KP} construct both a hopset (Definition \ref{def:hopset}) and a \emph{shortcut set} (Definition \ref{def:shortcut-set}). We use their constructions as a starting point. Recall that our hopset improves over their hopset, and matches the guarantee of their shortcut set (up to log factors). 

\paragraph{The Bottleneck of \cite{KP}.}

 First, we will outline the shortcut set and hopset constructions of \cite{KP}, and then we will identify the bottleneck of their hopset construction and the barrier towards extending it to match their shortcut set guarantee. 

The construction of the hopset of \cite{KP} is roughly as follows. Throughout the construction all edges $(u,v)$ added to the hopset $H$ have weight $\dist_G(u,v)$. First, they construct a collection of ``nice paths'' which are vertex-disjoint shortest paths on $\beta/12$ hops taken from the weighted transitive closure $G^*$ of the input graph (the nice paths also have some additional properties that we will not specify here). Then they add two different types of hopset edges: (1) hopset edges for each individual nice path, and (2) hopset edges from a random sample of vertices to a random sample of nice paths. Type (1) edges are where the bottleneck occurs, so we will focus on them. 

For type (1) edges, the goal is to add hopset edges to each nice path $P$ so that for any pair of vertices $x,y$ on $P$, there is a path on few hops from $x$ to $y$ of length $\dist(x,y)$. It is well known that one can add $O(|P|\log|P|)$ edges so that for all $x,y\in P$ where $x$ falls \emph{before} $y$ on $P$, there is a path on 2 hops of length $\dist(x,y)$ \cite{Ras10}. However, if $x$ falls \emph{after} $y$ on $P$, then this is no longer possible. 

This brings us to the previously mentioned major difference between shortcut sets and hopsets: for shortcut sets one can assume without loss of generality that the graph is a DAG. %(This is because for each strongly connected component one can simply add a bidirectional star and then contract the component.) 
Thus, for shortcut sets, for each nice path $P$ and vertices $x,y\in P$ where $x$ falls after $y$ on $P$, we know that $x$ cannot reach $y$ since the graph is a DAG, and thus we need not construct a small-hop path in the ``backward direction'' from $x$ to $y$; thus, for shortcut sets, adding $O(|P|\log(|P|)$ edges per nice path $P$ is sufficient.  On the other hand, for hopsets we cannot assume that the graph is a DAG and so we are faced with the challenge of constructing small-hop paths for \emph{all} pairs of vertices on $P$, including those in the backward direction. 

To get around this challenge, \cite{KP} simply add a hopset edge between \emph{every} pair of vertices on $P$. This strategy of adding $\Theta(|P|^2)$ hopset edges per path rather than $O(|P|\log |P|)$ comes at a cost that is reflected in the final bound. In fact, if we were able to add only $\tilde{O}(|P|)$ hopset edges per path, then we would immediately obtain a bound for hopsets that matches the shortcut set bound of \cite{KP}. Moreover, since hopsets allow for $(1+\eps)$-approximate distances, it would even be sufficient to add $\tilde{O}(|P|)$ hopset edges so that every pair of vertices $x,y\in P$ has a small-hop path of length $(1+\eps')\dist(x,y)$ for sufficiently small $\eps'$. However, this is \emph{provably impossible}. Hesse \cite{Hesse03} proved that there exists a DAG with no shortcut set (let alone a hopset) of size $\tilde{O}(n)$ with subpolynomial hopbound (see \cref{app} for the details of why this implies impossibility of our goal of obtaining a hopset for each nice path). Due to this impossibility,
%of obtaining for each $u,v\in P$ a small-hop path of length $(1+\eps')\dist(u,v)$, 
%it may appear that the approach of \cite{KP} is unsalvageable. 
it is not sufficient to simply develop a more refined shortcut set for each nice path P; we will need to change the larger framework of \cite{KP} as well.

\paragraph{General Strategy.} Our strategy for overcoming the above impossibility involves two main contributions. First, when adding hopset edges to each nice path $P$ to get a small-hop path from $x$ to $y$, we settle for approximating $\dist(x,y)$ to a factor that is sometimes, but not always, worse than the previously mentioned desired approximation factor of $(1+\eps')$. In particular, our first main contribution is a subroutine for adding hopset edges to each nice path $P$ that yields an approximation bound for $\dist(x,y)$ that depends on several different parameters that measure certain properties of $x$, $y$, and $P$. For some pairs $x,y$ on some nice path $P$, our small-hop path from $x$ to $y$ will have good approximation (the ``easy'' case), and we can take the small-hop path directly. On the other hand, for some pairs $x,y \in P$ our small-hop path from $x$ to $y$ will yield a bad approximation (the ``hard'' case), but we will show such hard cases always impose additional structure on $x$, $y$, and $P$. 
%When this approximation bound is bad for a particular pair $x,y\in P$, this means that these parameters of $x$, $y$, and $P$ have particular values, which imposes structure on the instance. 
Our second main contribution is a hierarchical sampling procedure whose purpose is to take advantage of this imposed structure of the hard case to obtain the final bound.

\paragraph{Contribution 1: Backward Shortcutting Subroutine.} 

We devise a ``backward shortcutting'' subroutine that yields for each nice path $P$ and each $x,y\in P$ a small-hop $xy$-path whose weighted length depends on several parameters. To understand these parameters, we need to introduce a couple of definitions. Let $G_{aug}$ be the graph augmented with the ``forward'' hopset edges for all of the nice paths i.e. for each path $P$ we add the set of $\tilde{O}(|P|)$ edges so that for all $x,y\in P$ where $x$ falls before $y$ on $P$, there is a path with 2 hops of length $\dist(x,y)$. The difficult case is thus when $x$ comes after $y$ on $P$. Now, we will define a canonical shortest path in $G_{aug}$ between each pair of vertices, that we call the \emph{road} (to distinguish it from the nice paths). For any pair $u,v$ of vertices, the \emph{road} $R(u,v)$ is a shortest $uv$-path in $G_{aug}$ with the fewest hops among all shortest $uv$-paths in $G_{aug}$. Observe that roads satisfy the standard substructure property: any subpath of a road is itself a road.

Now, we are ready to state the approximation bound of our backward shortcutting subroutine. The bound will depend on the following parameters:
\begin{itemize}[itemsep=0pt]
    \item $|P| = \tilde{\Theta}(\beta)$, which by construction will be exactly the same for each nice path,
    \item $\len(P)$: the weighted length of $P$,
    \item $h_P(y,x)$: the number of hops on $P$ from $y$ to $x$, and
    \item $|R(x,y) \cap P|$: the number of times $P$ intersects the road from $x$ to $y$.
\end{itemize}
The bound is as follows: after adding the hopset edges given by the backward shortcutting subroutine, for all $x,y\in P$ where $x$ appears after $y$ on $P$, there is an $xy$-path with at most 6 hops of length at most 

\[(1+\eps/2)\dist(x,y) + \frac{\eps\cdot\len(P)\cdot h_P(y,x)}{|P|\cdot|R(x,y) \cap P|}.\]

At first this bound appears quite difficult to parse, so we will give it some meaning. The first term $(1+\eps/2)\dist(x,y)$ is simply the error that we originally wished to achieve (for $\eps'=\eps/2$), but is in fact impossible. Thus, we will focus on the additional error introduced by the second term. The ratio $\len(P)/|P|$ is roughly the average weight of an edge on $P$. This is a weight normalization term that we will ignore for the sake of this overview. Finally, the ratio $h_P(y,x)/|R(x,y) \cap P|$ is the most conceptually important part of the bound. To understand this quantity, imagine that the subpath of $P$ from $y$ to $x$ hits the road $R(x,y)$, takes a detour, hits $R(x,y)$ again, takes another detour, and so on. The ratio $h_P(y,x)/|R(x,y) \cap P|$ is roughly the average number of hops on a detour. See \cref{fig:detour}.

\begin{figure}[h]
\centering
\includegraphics{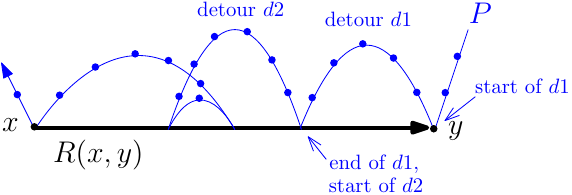}
\caption{Detours of the nice path $P$ with respect to $R(x,y)$. Note that we are in the difficult case where $x$ appears after $y$ on $P$.}
\label{fig:detour}
\end{figure}

Thus, roughly speaking, if $P$ takes small-hop detours with respect to $R(x,y)$, then the backward shortcutting subroutine provides a good approximation for $\dist(x,y)$ (the ``easy'' case mentioned in the general strategy above). On the other hand, if $P$ takes many-hop detours, then the backward shortcutting subroutine does not provide a good approximation (the ``hard'' case). Thus, we need a different technique to address the case of many-hop detours.

\paragraph{Contribution 2: Hierarchical Sampling for Many-Hop Detours.} We will take advantage of many-hop detours using the following simple fact: the more hops a detour has, the more likely we are to sample a vertex on the detour. 

Before understanding how our sampling procedure works, one first needs to understand an important difficulty that arises in the construction of our hopset that was not present in the hopset of \cite{KP}. Recall that the hopset of \cite{KP} includes edges from a random sample of vertices to a random sample of nice paths. The analysis of \cite{KP} is roughly as follows. Consider two arbitrary vertices $s,t$ and consider the road $R(s,t)$. Importantly, the road in \cite{KP} has a slightly different definition than our definition of the road because for them, $G_{aug}$ is the graph augmented with  \emph{all} $\Theta(|P|^2)$ hopset edges between all pairs of vertices in each nice path $P$, rather than just the hopset edges in the ``forward'' direction. Let $v_{first}$ be the sampled vertex on $R(s,t)$ that is closest to $s$, let $P_{last}$ be the sampled nice path that intersects $R(s,t)$ closest to $t$, and let $v_{last}$ be the vertex in $P_{last}\cap R(s,t)$ closest to $t$. See \cref{fig:kp}. They argue that (1) $v_{first}$ is few hops from $s$, and (2) $v_{last}$ is few hops from $t$.
%(1) their random sample of vertices includes a vertex $v_{first}$ near the beginning of $R(s,t)$, and (2) their random sample of nice paths contains a path $P_{last}$ that intersects $R(s,t)$ at a vertex $v_{last}$ near the end of $R(s,t)$. 
They can thus construct a small-hop path from $s$ to $t$ as follows: follow $R(s,t)$ from $s$ to $v_{first}$, then take the hopset edge from $v_{first}$ to $P_{last}$, then take the hopset edges along $P_{last}$ to reach $v_{last}$, and finally follow $R(s,t)$ from $v_{last}$ to $t$. %First taking $R(s,t)$ from $s$ to a sampled vertex $v$ near the beginning of $R(s,t)$, then taking the hopset edge from $v$ to a sampled path that intersects near the end of $R(s,t)$, and finally taking the remainder of $R(s,t)$ to $t$. 

\begin{figure}[h]
\centering
\includegraphics{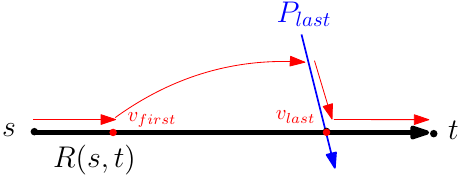}
\caption{The $st$-path in the hopset of \cite{KP}.}
\label{fig:kp}
\end{figure}

The difficulty for our hopset is (2) is no longer true. That is, using the above notation, we cannot guarantee that $v_{last}$ is few hops from t.
%That is, our random sample of nice paths may \emph{not} contain a path that intersects near the end of $R(s,t)$.
%That is, using the above notation, we cannot guarantee that $v_{last}$ is few hops from t. %\nicole{Previously I had a description next for why (2) is true for KP, but removed it to get to the point faster. It is now commented out below, do you think it should be included?} 
%To see why, we will first briefly describe why (2) is true in the hopset of \cite{KP}. First, they observe that each nice path can only intersect $R(s,t)$ at most twice. This is because if a nice path $P$ intersected $R(s,t)$ more than twice, we could take one of the $\Theta(|P|^2)$ hopset edges directly from the first intersection to the last intersection to get a shortest $st$-path with fewer hops, which contradicts the extremal choice of the road $R(s,t)$. Then, they consider a short suffix of $R(s,t)$ and observe that because almost every vertex of $R(s,t)$ is on some nice path, and each nice path intersects $R(s,t)$ at most twice, this suffix intersects many different nice paths, and thus a random sample of nice paths hits this suffix.
%The above argument breaks for our hopset because a nice path can intersect $R(s,t)$ \emph{arbitrarily many} times, instead of just twice. 
This stems from the fact that in the construction of \cite{KP}, having all $\Theta(|P|^2)$ hopset edges implies that a nice path can only intersect $R(s,t)$ at most twice, while in our construction, a nice path can intersect $R(s,t)$ \emph{arbitrarily many} times. Even though it is not true in our hopset that  $v_{last}$ is among the last vertices of $R(s,t)$,
%in our hopset a sampled nice path will hit among the last vertices of $R(s,t)$,
 a trivial probability analysis shows that it \emph{is} true that $P_{last}$ is one of the last nice paths to intersect $R(s,t)$; specifically, if nice paths are sampled with probability $1/\beta$, then by definition of $P_{last}$ and $v_{last}$ there are (in expectation) at most ${\sim}\beta$ nice paths that intersect $R(v_{last}, t)$. 
 %that a sampled nice path will hit among the last nice paths that hit $R(s,t)$, specifically the last  ${\sim}\beta$ nice paths. 
 It is not clear a priori how this fact could be useful. We will take advantage of this fact in a somewhat nuanced way. %This brings us to the idea of hierarchical sampling.

Since sampling a nice path among the last ${\sim}\beta$ nice paths to intersect $R(s,t)$ does not help us directly, we would instead like to sample a \emph{larger} set of nice paths, so that $P_{last}$ will hit closer to the end of $R(s,t)$. However, since we use the strategy of \cite{KP} of adding hopset edges between the sampled vertices and the sampled paths, if we want to maintain the size of the hopset then increasing the number of sampled paths requires \emph{decreasing} the number of sampled vertices. However, we cannot afford to simply sample fewer vertices because this would (in expectation) increase the number of hops from $s$ to $v_{first}$.  

To overcome this issue, we sample a hierarchy of $\log n$ levels of both vertices and nice paths, where in higher levels we sample \emph{more} paths and \emph{fewer} vertices. Then we add hopset edges between the sampled vertices in each level to the sampled nice paths in the same level. Next, we will explain why this method of hierarchical sampling helps. 

See \cref{fig:path} for a depiction of the path we build from $s$ to $t$, which we describe next. First, just like in \cite{KP}, a vertex $v_{first}$ sampled at level 0 hits near the beginning of $R(s,t)$. We take the hopset edge from $v_{first}$ to a path sampled at level 0 that falls among the last ${\sim}\beta$ nice paths to hit $R(s,t)$. Our next goal is to find a small-hop path from here to a vertex sampled at level 1, which allows us to take a hopset edge to a path sampled on level 1. Because twice as many paths are sampled on level 1 than level 0, a path sampled on level 1 falls among the last $\sim\beta/2$ nice paths to hit $R(s,t)$, and so we have made progress towards $t$. We iterate this strategy: in general if we are currently at a path sampled on level $i$, our next goal is to find a small-hop path to a vertex sampled on level $i+1$, which allows us to take a hopset edge to a path sampled on level $i+1$, which falls among the last $\sim\beta/2^i$ nice paths to hit $R(s,t)$.

\begin{figure}[h]
\centering
\includegraphics{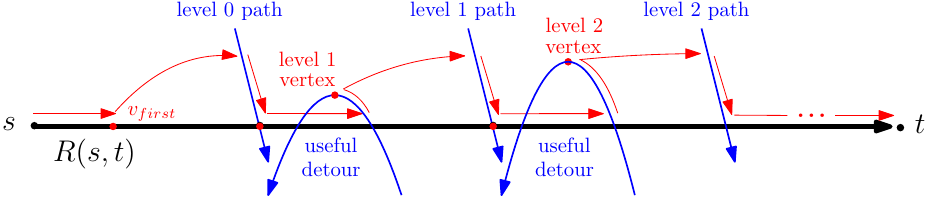}
\caption{The $st$-path in our hopset.}
\label{fig:path}
\end{figure}

Thus, the question is: how do we get from a nice path sampled on level $i$ to a vertex sampled on level $i+1$ via a small-hop path? Specifically, our starting point is the vertex $v_{last}$ defined as the last vertex on $R(s,t)$ that is on a nice path sampled on level $i$. The key idea is that because only ${\sim}\beta/2^i$ nice paths intersect $R(v_{last},t)$, the larger the level $i$ is, the more error we can afford to incur per such nice path. That is, the previously mentioned ``easy case'' (when the backward shortcutting subroutine provides a good approximation) has a higher error threshold when the level $i$ is larger. This means that in the ``hard case'' these ${\sim}\beta/2^i$ nice paths take detours of \emph{even more} hops when the level $i$ is larger.  
%This means that it is acceptable if the guarantee of our backward shortcutting subroutine is worse for these ${\sim}\beta$ nice paths. Thus, if the guarantee is not acceptable, then the aforementioned \emph{detours} of these nice paths must contain more hops per detour. 

As previously mentioned, the benefit of having many-hop detours is that it is easier to sample a vertex on such a detour. However, there is a caveat: not all detours are useful to sample from. Because we need a \emph{small-hop} path to a vertex sampled on level $i+1$, the only useful detours are those that start only few hops from $v_{last}$. See \cref{fig:detour} for a depiction of the start of a detour. Such detours are useful because if there is a vertex $v$ sampled on level $i+1$ on such a detour, we can take few hops from $v_{last}$ to the start of the detour, and then take the hopset edges along the detour path to reach $v$ in few hops.

To handle the fact that we only care about useful detours, our analysis takes the form of iterative walk down $R(s,t)$, where at each step we compare the number of steps we have taken with the number of vertices on detours whose starting point we have reached. For instance, if we walk 10 steps down $R(s,t)$ and hit the starting point of detours whose combined number of hops is 100, then we have gained a sampling advantage because sampling any of these 100 vertices would be useful. Generally, we show that the larger the level $i$, the more sampling advantage we gain per step of our walk. That is, from above we know that larger $i$ means that we have many-hop detours (otherwise we would be in the ``easy case''), but we additionally show via a more fine-grained analysis, that larger $i$ also means that we have many-hop \emph{useful} detours. This allows us to conclude that for larger $i$, a smaller sample of vertices is needed in order to hit a vertex on a useful detour. This is why we can afford to sample fewer vertices at higher levels in our hierarchical sampling scheme.

\paragraph{Remaining Details.}

We have described the main ideas behind the construction, however there are several additional issues that need to be addressed in the full proof. For instance, we need to bound the (weighted) length of each nice path for two reasons: (1) we need to bound the length of each detour to obtain the desired distance approximation, and (2) we need to bound the term $\len(P)/|P|$ in the guarantee of the backward shortcutting subroutine. We bound the length of each nice path in a similar spirit to \cite{KP}: by only considering nice paths of bounded length and completely ignoring the rest of the nice paths. Specifically, when performing the analysis for a pair of vertices $s,t$, we only consider the ``relevant'' nice paths whose length is small enough as a function of $\dist(s,t)$.

Additionally, we need to bound the length of the segment of $R(s,t)$ covered by each detour to ensure the desired distance approximation. To do so, if some such segment of $R(s,t)$ has very large length we handle it in a brute-force manner; we can afford to do so because there can only be few segments of very large length.

\section{Proof of \cref{thm:main}}

We will spend almost the entire proof focusing on the first bound listed in \cref{thm:main} for $\beta\leq n^{1/3}$. The second bound (for $\beta>n^{1/3}$) follows from the first after applying a black-box transformation of \cite{KP}; See Section \ref{sec:large-beta} for details.

\subsection{Hopset Construction}

Throughout the construction all edges $(u,v)$ added to the hopset $H$ have weight $\dist_G(u,v)$. 

\subsubsection{Construction Part 1: Using \cite{KP} as a Starting Point}\label{sec:part1}

In this section, we will describe the first part of our construction, which comes directly from Kogan and Parter's construction \cite{KP}.
First we pick a \emph{nice path collection}. 

%, and then form a \emph{refined path collection}.

\paragraph{Nice Path Collection.}

\begin{definition}[Nice Path Collection \cite{KP}]

An collection $\mathcal{P}=\{P_1,\dots,P_k\}$ of vertex-disjoint paths is \emph{nice} if each $P_i$ satisfies the following properties:
\begin{enumerate}[label=(P\arabic*)]
    \item\label{prop1} $P_i\in G^*$ (the weighted transitive closure of $G$),
    \item\label{prop2} $P_i$ has exactly $\eps\beta/(162\log n)$ hops, 
    \item\label{prop3} $P_i$ is a shortest path in $G^*$,
    \item\label{prop4} $\len(P_i)$ is the shortest among all other $\eps\beta/(162\log n)$-hop shortest paths in $G^*\setminus \cup_{j<i}V(P_j)$,
\end{enumerate}
and $\mathcal{P}$ is maximal i.e. there is no path respecting properties \ref{prop1}-\ref{prop4} that can be added to $\mathcal{P}$.
\end{definition}
Our definition of a nice path collection is identical to that of \cite{KP} except for the value of the parameter in (P2) indicating the number of hops on each path.

% Now we define the \emph{refined path collection} by partitioning each path of the nice path collection into subpaths.

% \begin{definition}[Refined Path Collection \cite{KP}]
% Given a nice path collection $\mathcal{P}=\{P_1,\dots,P_k\}$, obtain a \emph{refined path collection} $\mathcal{P}'$ by partitioning the vertices of each $P_i$ into $O(1/\eps)$ subpaths $P_{i,1},\dots,P_{i,k_i}$, each of length at most $\eps\cdot\len(P_i)$.
% \end{definition}

% For the remainder of the construction, we will only work with the refined path collection $\mathcal{P}'$ (rather than the nice path collection).

For each path $P\in \mathcal{P}$, we add to our hopset $H$ a set $H(P)$ of $O(|P|\cdot\log|P|)$ edges such that for all vertices $u,v\in P$ where $u$ appears before $v$ on $P$, there is a path in $H(P)\cup P$ on only 2 hops whose length is $\dist_P(u,v)=\dist_G(u,v)$. These edges $H(P)$ can be computed by a known procedure (see e.g. \cite{Ras10}).% and were also utilized in \cite{KP} (in their shortcut set construction, not their hopset). 

% We note that in \cite{KP}, they add these edges $H(P)$ in their \emph{shortcut set} construction, but for their hopset construction they add more edges, which we do not add. In particular, in their hopset construction they add an edge between \emph{every} pair of vertices on $P_i$. Adding these $\Omega(|P|^2)$ edges is the expensive step of their hopset construction, and our goal is to avoid adding this many edges.

\paragraph{Vertex-Path Hopset Procedure.}
In \cite{KP}, they define a procedure for adding a set of hopset edges from any vertex $v$ to any path $P\in \mathcal{P}$. We will call this procedure the \emph{vertex-path hopset procedure}. The guarantee they obtain is the following: One can add a set $H(v,P)$ of $O(\log(nW)/\eps)$ edges so that for each vertex $p\in P$, there is a path in $H(v,P) \cup H(P) \cup P$ from $v$ to $p$ with at most 3 hops and length at most $(1+\eps/2)\dist(v,p)+\len(P)$ (see Eqn. 3.1 in \cite{KP}). 

We will use the vertex-path hopset procedure to add sets of hopset edges from vertices to paths in $\mathcal{P}$. However the specific vertices and paths that we choose will be different from those in \cite{KP}, and are defined in the next section.

\subsubsection{Construction Part 2: New Hopset Edges}\label{sec:part2}

In this section, we will introduce our new hopset edges that do not appear in $\cite{KP}$. They come in two types: (1) hopset edges between sampled vertices and sampled nice paths, and (2) hopset edges on each nice path. 

\paragraph{Vertex-Path Hopset Edges.} We define hierarchical levels of sampling, where at each level we sample a subset of vertices and a subset of the nice paths from $\mathcal{P}$ and add hopset edges between them. In contrast, \cite{KP} used only one level of sampling. This hierarchical sampling procedure and its analysis is one of our two main contributions, as outlined in the technical overview. For each level $i$ from 1 to $\log n$, we sample: \[\text{each vertex independently with probability } \Theta\Big(\frac{\log^2 n}{2^i\beta}\Big)\] \begin{center}and\end{center} \[\text{each nice path from $\mathcal{P}$ independently with probability } \Theta\Big(\frac{2^i\cdot \log^3 n}{\beta}\Big).\] These samples are nested so that the vertices sampled on level $i$ are a subset of those sampled on level $i-1$, and the paths sampled on level $i$ are a superset of those sampled on level $i-1$. 

We note that these sampling probabilities imply that with high probability the number of sampled vertices is $O(n\log^2 n/(2^i\beta))$, and 
the number of sampled nice paths is $O(2^i\cdot n\log^4 n/(\eps\beta^2))$ since the total number of nice paths is at most $|\mathcal{P}|=O(n\log n / (\eps\beta))$. (For the concentration bounds, recall that $\beta \leq n^{1/3}$, so $n/\beta^2 \gg \log(n)$.)

Using the aforementioned vertex-path hopset procedure, for each $i$ we add hopset edges from each sampled vertex on level $i$ to each sampled path on level $i$.
%i.e. a $\Theta(2^i\cdot \log n/\beta)$ fraction of the paths.

\paragraph{Hopset Edges for Nice Paths.} So far we have only added hopset edges to each nice path $P_i$ to ensure that there are small-hop shortest paths from each vertex $u$ to each vertex $v$ that appears \emph{after} $u$ on $P_i$. The purpose of our additional hopset edges is to create small-hop paths in the \emph{backward} direction (from $v$ to $u$).

We will define a \emph{backward shortcutting subroutine} that we will apply to each nice path. This subroutine is one of our two main contributions, as outlined in the technical overview.
% Ideally, we would like the guarantee that there is a small-hop path from $v$ to $u$ of length $(1+\eps)\dist(v,u)$, however a hopset of small size with that guarantee does not exist. \nicole{Todo: include example for why it does not exist in the technical overview}. Instead, our backward shortcutting subroutine
%The distance guarantee of the subroutine depends on the value of several variables. 
To state the approximation guarantee of the subroutine, we need to define the following terminology.

\begin{notation}
Let $G_{aug}$ be the graph $G$ augmented with the previously defined hopset edges $H(P)$ for all $P\in \mathcal{P}$, that is, $G_{aug}=(V,E\cup (\cup_i H(P_i)))$.
\end{notation}

We now define a canonical shortest path between each pair of vertices, known as the \emph{road} (to distinguish it from the previously defined nice paths).

\begin{definition}[Road]
Given vertices $u,v$, the \emph{road} of $u,v$ denoted $R(u,v)$ is defined as a shortest $uv$-path in $G_{aug}$ with the fewest hops among all shortest $uv$-paths in $G_{aug}$. Ties are broken consistently so that if $w,x$ are vertices on $R(u,v)$ with $w$ before $x$, then $R(w,x)$ is a subpath of $R(u,v)$.
\end{definition}

Now we are ready to state the guarantee of our backward shortcutting subroutine. Recall that for a pair of vertices $u,v$ on a shortest path $P$ where $u$ appears before $v$ on $P$, $h_P(u,v)$ denotes the number of hops on $P$ from $u$ to $v$. 

\begin{restatable}[Backward Shortcutting Subroutine]{lemma}{lemback}\label{lem:back}
Given a nice path $P\in\mathcal{P}$ and $\gamma,\delta\in(0,1)$, it is possible to add a set of \[O\Big(\frac{|P| \cdot \log^2(nW)\log^3 n}{ \gamma\delta}\Big)\] hopset edges to $G_{aug}$ so that for all vertices $x,y\in P$ where x appears after y on $P$, the resulting graph has an $xy$-path with at most 6 hops and length at most 
\[(1+\gamma)\dist(x,y) + \frac{\delta\cdot\len(P)\cdot h_P(y,x)}{|P|\cdot|R(x,y) \cap P|}.\]
\end{restatable}

% To give the second term of the guarantee of \cref{lem:back} more meaning, we briefly break down its parts: The ratio $\len(P)/|P|$ is roughly the average weight of an edge on $P$. The ratio $h_P(y,x)/|R(x,y) \cap P|$ is roughly the fraction of the subpath of $P$ from $y$ to $x$ that is hit by road from $x$ to $y$.

To complete the construction of our hopset, we apply the backward shortcutting subroutine to every nice path $P\in \mathcal{P}$ with $\gamma=\eps/2$ and $\delta=\eps$.

The remainder of the paper is divided into two sections. First, we use the backward shortcutting subroutine guarantee from \cref{lem:back} as a black box to prove \cref{thm:main}. After that, we prove \cref{lem:back}.

\subsection{Proof of \cref{thm:main} using Backward Shortcutting as a Black Box}

\subsubsection{Hopset Size Analysis} 

We will calculate the number of edges in our hopset $H$. Recall that \cref{thm:main} claims the following bound: $|H|=O(n^2 \log^7 n \log^2(nW)/(\eps^2 \beta^3))$.  There are three types of hopset edges: \begin{enumerate}[itemsep=0pt]
    \item those in $H(P)$ for some $P\in \mathcal{P}'$ (specified in \cref{sec:part1}),
    \item vertex-path hopset edges (specified in \cref{sec:part2}), and 
    \item hopset edges from the backward shortcutting subroutine (\cref{lem:back}).
\end{enumerate}

The number of hopset edges of type 1 for each path $P\in \mathcal{P}$ is $O(|P|\cdot\log|P|)$. Since paths in $\mathcal{P}$ are pairwise (vertex) disjoint, this is a total of $O(n\log n)$ edges.

Now, we calculate the number of hopset edges of type 2. For each level $i$, we add hopset edges from $O((n\log^2 n)/(2^i\beta))$ vertices to $O((2^i\cdot n\log^4 n)/(\eps\beta^2))$ paths (with high probability). From the guarantee of the vertex-path hopset procedure, with high probability the number of hopset edges for each vertex-path pair is $O(\log(nW)/\eps)$. Thus, with high probability the number of hopset edges added for each level $i$ is \[O\Big(\frac{n^2\log^6 n \log(nW)}{\eps^2 \beta^3}\Big).\] There are $\log n$ levels $i$, so we multiply by $\log n$ to get the total number of hopset edges of type 2.

Finally, we calculate the number of hopset edges of type 3. By the guarantee of the backward shortcutting subroutine (\cref{lem:back}), there are $O(|P| \cdot \log^2(nW)\log^3 n/\eps^2)$ hopset edges of type 3 for each $P\in \mathcal{P}'$. Since paths in $\mathcal{P}'$ are pairwise (vertex) disjoint, this is a total of $O(n\log^2(nW)\log^3 n/\eps^2)$. %\nicole{Todo: make sure $\gamma$ and $\delta$ are indeed both $\Theta(\eps)$. Probably $\delta$ will have to be $\eps^2$ actually?}

\subsubsection{Constructing the Path} Fix vertices $s,t\in V$. In this section we will construct a path from $s$ to $t$ in $G\cup H$ that we will later argue satisfies the hopbound and distance approximation bound claimed in \cref{thm:main}.

The analysis will focus on the road $R(s,t)$. If $R(s,t)$ already has at most $\beta$ hops then we are done, so assume otherwise. 

% Let $j$ be the maximum value such that at least $\iota\beta$ vertices on $R(s,t)$ are \emph{not} contained in any of the first $j$ paths $\{P_1,\dots,P_j\}$ in the nice path collection. \nicole{todo set $\iota$, it will be like $\eps/2$ or something.} We say that a nice path is \emph{relevant} if it is in $\{P_1,\dots,P_j\}$. %For the sake of analysis, we will use a trick introduced in \cite{KP}, which is to \emph{only} consider the nice paths $\{P_1,\dots,P_j\}$, ignoring the rest of the paths. %We say that a refined path is \emph{relevant} if it is a subpath of some nice path in $\{P_1,\dots,P_j\}$.  
% \nicole{or we could define relevant just based on the length and then argue thater that there are only few uncovered vertices. This sounds better.}

In a similar spirit to $\cite{KP}$, we will only consider a subset of the nice paths in our analysis, completely disregarding the rest. In particular, we define a \emph{relevant nice path} as follows.

\begin{definition}[Relevant Nice Path]
A nice path $P$ is \emph{relevant} if $\len(P)\leq \eps\dist(s,t)/(54\log n)$.
\end{definition} 
%In the analysis we will only include consider relevant nice paths, completely disregarding all other nice paths.

Let $Q$ be the path that we will build from $s$ to $t$. We will build $Q$ in phases. In the $i^{th}$ phase, for $i$ from 0 to $\log n$, we will build a subpath $Q_i$. The final path $Q$ will be the concatenation of the subpaths $Q_i$. 

To define each $Q_i$, We will first define vertices $s_0,\dots, s_{1+\log n}\in R(s,t)$. We will impose the property that each $Q_i$ starts at a vertex that is at least as far along $R(s,t)$ as $s_i$ (i.e. $Q_i$ starts at a vertex $u$ such that $u\in R(s,t)$ and $\dist(u,t)\leq \dist(s_i,t)$), and $Q_i$ ends at a vertex at least as far along $R(s,t)$ as $s_{i+1}$. The last vertex on $Q_i$ is the first vertex on $Q_{i+1}$. Define these vertices $s_i$ as follows.

\begin{notation}
Let $s_0=s$, let $s_{1+\log n}=t$, and let each remaining $s_i$ be the last vertex on $R(s,t)$ belonging to a relevant nice path that is sampled at level $i$. 
\end{notation}
Note that the $s_i$'s appear in order on $R(s,t)$ since the samples of nice paths are nested.

We first define $Q_i$ for $i=0$ and $i=\log n$. $Q_{\log n}$ is simply defined as $R(s_{\log n},t)$. Let $s'$ be the first vertex on $R(s,t)$ that is sampled on level 1. $Q_0$ is simply defined as $R(s,s')$ concatenated with the path from $s'$ to $s_1$ given by the vertex-path hopset procedure.

Now, we will define $Q_i$ for the remaining values of $i$. Fix $i\not=0,\log n$ and suppose we are at the beginning of phase $i$. That is, we are about to build a path $Q_i$ starting at a vertex at least as far along $R(s,t)$ as $s_i$.

We will define certain intervals of $R(s,t)$ as \emph{easy} or \emph{hard}. Intuitively, the easy intervals are those for which the backward shortcutting subroutine from \cref{lem:back} provides a good distance approximation. Formally, we define the easy and hard intervals as follows:

\begin{definition}[Easy/Hard Intervals]\label{def:easyhard}
Let $P$ be a relevant nice path, and let $x$ and $y$ be arbitrary vertices in $R(s,t)\cap P$ (if such vertices exist) where $x$ appears before $y$ on $R(s,t)$. We say that the interval $[x,y]$ of $R(s,t)$ is \emph{easy} if either $x$ falls before $y$ on $P$, or $h_P(y,x)/|R(x,y) \cap P| < 2^i$. Otherwise $[x,y]$ is a \emph{hard} interval of $R(s,t)$. 
\end{definition}

Starting at $s_i$, we will construct $Q_i$ by walking along $R(s,t)$ step-by-step, and at each step we will decide whether or not to take a path outside of $R(s,t)$ to shortcut directly to a vertex later on $R(s,t)$. Specifically, we perform the following loop.\\

\begin{breakablealgorithm}
\begin{algorithmic} 
\floatname{algorithm}{}
\caption{\hspace{-2mm}{\bf Algorithm for constructing $Q_i$}}
\begin{tcolorbox}[
  blanker,
  breakable
  ]
%I did the following part manually and added a dummy item to the itemize environment below because the enumerate/itemize environments put the first two numbers/bullets on top of each other for some reason. This might be hackiest thing I have ever done in latex.
Starting from $s_i$, do the following until the phase terminates:
 \vspace{1mm}
 
   \hspace{2mm} 1. \hspace{1mm}Take a step forward on $R(s,t)$ and let $v$ be the vertex stepped to. 
   \vspace{1mm}
   
   \hspace{2mm}  2. \hspace{1mm}Append $v$ to $Q_i$. If $v=s_{i+1}$, then end the phase.
   \vspace{1mm}
   
   \hspace{2mm}  3. \hspace{1mm}Go through the following cases, returning to step 1 if a case finishes without terminating the phase:

\begin{itemize}
\item $ $
\vspace{-5mm}
    \item If $v=t$, terminate the algorithm.
    \item If $v$ falls on a relevant nice path $P$ and $v$ is the \emph{first} vertex on $Q_i$ that is on $P$, then do the following: Let $v_f$ be the final vertex on $R(s,t)$ that is on $P$. (That is, among all vertices on $R(s,t) \cap P$, $v_f$ is the one closest to $t$.)
    \begin{itemize}
        \item \hypertarget{s1}{(S1)} If $\dist(v,v_f) > \eps\dist(s,t)/(9 \log n)$, then we append a path from $v$ to $v_f$ onto $Q_i$ as follows. If $v$ appears before $v_f$ on $P$ then take the 2-hop path given by $H(P)$, and if $v$ appears after $v_f$ on $P$ then take the path given by applying the backward shortcutting subroutine (\cref{lem:back}) to $P$. If $v_f$ is at least as far along $R(s,t)$ as $s_{i+1}$, then end the phase; otherwise, go back to step 1 with $v_f$ as the current position of $Q_i$. 
        \item \hypertarget{s2}{(S2)} Otherwise, let $v_{easy}$ be the last vertex on $R(s,t)$ such that $v_{easy}$ is on $P$ and $[v,v_{easy}]$ is an easy interval. (The vertex $v_{easy}$ always exists since $v$ is a valid choice for $v_{easy}$). Append a path from $v$ to $v_{easy}$ onto $Q_i$ as follows. If $v$ appears before $v_{easy}$ on $P$, then take the 2-hop path from $v$ to $v_{easy}$ given by $H(P)$. If $v$ appears after $v_{easy}$ on $P$, then take the path from $v$ to $v_{easy}$ given by applying the backward shortcutting subroutine (\cref{lem:back}) to $P$. If $v_{easy}$ is at least as far along $R(s,t)$ as $s_{i+1}$, then end the phase;  otherwise, go back to step 1 with $v_{easy}$ as the current position of $Q_i$. 
    \end{itemize}
    \item If $v$ falls on a relevant nice path $P$ and $v$ is not the first vertex on $Q_i$ that is on $P$, then do the following: Let $v_0$ be the first vertex on $Q_i$ that is on $P$. 
    \begin{itemize}
      \item \hypertarget{s3}{(S3)} If $v$ falls before $v_0$ on $P$ and there is a vertex $u$ sampled on level $i+1$ that falls between $v$ and $v_0$ (inclusive) on $P$, then append a path from $v$ to $s_{i+1}$ onto $Q_i$ as follows. Take the 2-hop path from $v$ to $u$ on $P$ given by $H(P)$. From there, take the path from $u$ to $s_{i+1}$ given by the vertex-path hopset procedure. End the phase.
      \item Otherwise, do nothing and go back to step 1.
    \end{itemize}

\item Otherwise (if $v$ does not fall on a relevant nice path), do nothing and go back to step 1.
\end{itemize}

\end{tcolorbox}

\end{algorithmic}
\end{breakablealgorithm}

%We observe that if we start an iteration of the loop at vertex $v$, we end that iteration at a vertex on $R(s,t)$ that is at least as close to $t$ as $v$ is. (This is deterministically true even for \hyperlink{s3}{(S3)} since the samples of nice paths are nested, so the $s_i$'s appear in order on $R(s,t)$.) That is, over time we only make positive progress along $R(s,t)$. 

The following claim will be useful in the analysis:

\begin{claim}\label{claim:first}
    Let $z$ be the \emph{first} vertex on $Q_i$ (for some $i$) from a given relevant nice path $P$. If the loop is run on another vertex $z'\in Q_i\cap P$, then the loop is run on $z$.
\end{claim}
\begin{proof}
    When the loops is run on $z'$, no vertex on $Q_i$ has yet fallen into case (S3) because this case immediately ends the phase. The only vertices before $z'$ on $Q_i$ that we do not run the loop on, are vertices appended to $Q_i$ during the loop in case (S1) or (S2). In both cases, vertices appended during the loop are always on the same relevant nice path as the vertex we are executing the loop on. Thus, $z$ cannot be such an appended vertex since $z$ is the \emph{first} vertex on $Q_i$ that is on $P$.
\end{proof}

% Note that the above loop is executed on every vertex that is the \emph{first} vertex on $Q_i$ from a given relevant nice path $P$, for the following reason. The only vertices on $Q_i$ that we do not run the loop on, are vertices appended to $Q_i$ during the loop. And vertices appended during the loop are always on the same relevant nice path as the vertex $v$ we are executing the loop on. We will use this fact implicitly in the below analysis.

\subsubsection{Distance Approximation Analysis}

% First, we establish a relationship between $\dist(s,t)$ and the length of the relevant refined paths:  

% \begin{claim}\label{claim:reflen}
% For each relevant refined path $P$, $\len(P)\leq\eps\cdot\dist(s,t)$.
% \end{claim}

% \begin{proof}
% Fix a relevant refined path $P$, and let $i$ be such that $P_i$ is the nice path for which $P$ is a subpath. Since $P$ is relevant, we know that $i<j$ so at least $\beta/c+1$ vertices on $R(s,t)$ are not contained in any of the first $i-1$ paths $\{P_1,\dots,P_{i-1}\}$ in the nice path collection. Because these $\beta/c+1$ vertices are on $R(s,t)$, they form a shortest path $P'$ in $G^*$ with $\beta/c$ hops whose length is at most $\dist(s,t)$. Thus, $P'$ satisfies properties \ref{prop1}-\ref{prop3} of the nice path collection and is thus a candidate to be chosen as $P_i$. By property \ref{prop4} of the nice path collection, $\len(P_i)$ is the shortest among all candidates for $P_i$. Therefore, $\len(P_i)\leq \len(P')\leq\dist(s,t)$. Then, since $\len(P)\leq \eps\cdot \len(P_i)$ by the definition of a refined path, the claim follows.
% \end{proof}

For each source of error, we will calculate both its multiplicative and additive error. We will argue that the multiplicative error for any segment of the path $Q$ is at most $1+\eps/2$, while the sum of the additive error over all sources is at most $\eps\dist(s,t)/2$.

First, we will analyze the distance error for $Q_0$, which is defined as the concatenation of (1) the path $R(s,s')$ which is an exact shortest path, and (2) the path from $s'$ to $s_1$ given by the vertex-path hopset procedure, which guarantees that $\len(Q_0)\leq (1+\eps/2)\dist(s',s_1)+\len(P)$ where $P$ is the relevant nice path containing $s_1$. By the definition of a relevant path, $\len(P)\leq \eps\dist(s,t)/(54\log n)$. That is, $Q_0$ has multiplicative error $1+\eps/2$ and additive error $\eps\dist(s,t)/(54\log n)$.

Next, we will analyze the distance error induced by each of the three types of shortcuts \hyperlink{s1}{(S1)}-\hyperlink{s3}{(S3)} that our path $Q$ takes. This is the only additional source of error since $Q_1,\dots, Q_{\log n}$ simply follow $R(s,t)$ when such a shortcut is not taken.

\paragraph[hi]{Distance Error from \hyperlink{s1}{(S1)}.}  If $v$ falls before $v_f$ on $P$, we use the 2-hop path from $v$ to $v_f$ given by $H(P)$, which preserves the exact distance from $v$ to $v_f$. On the other hand, if $v$ falls after $v_{easy}$ on $P$, then we use the backward shortcutting subroutine (\cref{lem:back}), which implies that the length of the subpath of $Q_i$ from $v$ to $v_f$ is at most $(1+\eps/2)\dist(v,v_f)+\eps\cdot \len(P)$. That is, the multiplicative error is $1+\eps/2$, and the additive error is at most $\eps\cdot \len(P)$. By the definition of a relevant path, $\len(P)\leq \eps\dist(s,t)/(54\log n)$, so the additive error is at most $\eps\cdot \len(P)\leq \eps^2\cdot\dist(s,t)/(54\log n)$. 

By the specification of \hyperlink{s1}{(S1)}, $\dist(v,v_f) > \eps\dist(s,t)/(9 \log n)$, so when going from $v$ to $v_f$ we progress more than an $\eps/(9 \log n)$ fraction down the road $R(s,t)$. This can only happen at most $(9 \log n)/\eps$ times in total throughout the entire process of building the path $Q$. Thus, the total amount of additive error from \hyperlink{s1}{(S1)} is at most $\eps\dist(s,t)/6$. 

\paragraph[hi]{Distance Error from \hyperlink{s2}{(S2)}.} If $v$ falls before $v_{easy}$ on $P$, we use the 2-hop path from $v$ to $v_{easy}$ given by $H(P)$, which preserves the exact distance from $v$ to $v_{easy}$. On the other hand, if $v$ falls after $v_{easy}$ on $P$, then since $[v,v_{easy}]$ is an easy interval, by \cref{def:easyhard} we know that \begin{equation}\label{eqn:easy}\frac{h_P(v_{easy},v)}{|R(v,v_{easy}) \cap P|} < 2^i. \end{equation}

By the guarantee of the backward shortcutting subroutine (\cref{lem:back}), the length of the subpath of $Q_i$ from $v$ to $v_{easy}$ is at most 
\[
    (1+\eps/2)\dist(v,v_{easy}) + \frac{\eps\cdot\len(P)\cdot h_P(v_{easy},v)}{|P|\cdot|R(v,v_{easy}) \cap P|}.\]
    We will simplify the second term (the additive error):
    \begin{align*}
    \frac{\eps\cdot\len(P)\cdot h_P(v_{easy},v)}{|P|\cdot|R(v,v_{easy}) \cap P|}&<  \frac{\eps\cdot\len(P)\cdot 2^i}{ |P|} &\text{by \cref{eqn:easy}}\\
    &= \frac{\eps\cdot\len(P)\cdot 2^i}{ \eps\beta/(162\log n)}&\text{by the construction of nice paths}\\
    &\leq\frac{3\eps\cdot 2^i}{\beta}\dist(s,t)&\text{by the definition of a relevant path.}\\
\end{align*}

Now we will calculate how many times we can fall into case \hyperlink{s2}{(S2)}. We fall into this case at most one per relevant nice path $P$ that intersects $R(s_i,t)$, since we only reach case \hyperlink{s2}{(S2)} if $v$ is the \emph{first} vertex on $Q_i$ that is on $P$. Thus, we will calculate how many relevant nice paths can intersect $R(s_i,t)$. 
%First, we will calculate $|R(s_i,t)|$. The number of vertices sampled at level $i$ is $\Theta(n\log n/(2^i\beta)$. Thus, with high probability, some vertex sampled at level $i$ hits within the last $O(2^i\beta)$ vertices on $R(s,t)$. 

\begin{claim}\label{claim:numpaths}
For any $i$ from 1 to $\log n$, with high probability the number of relevant nice paths intersecting $R(s_i,t)$ is at most $24\beta/( 2^i\log^2 n)$.
\end{claim}

\begin{proof}
% The number of nice paths sampled at level $i$ is $\Theta(2^in\log^4 n/({\eps\beta^2}))$. The total number of nice paths is at most $24\cdot n\log n/(\eps\beta)$ since each nice path has $\eps\beta/(24\log n)$ hops and they are all disjoint. Thus, on level $i$ each nice path (and thus, each relevant nice path) is sampled with probability $\Omega(2^i\log^3 n/(24\beta))$. 

The vertex $s_i$ is defined as the last vertex on $R(s,t)$ belonging to a relevant nice path that is sampled at level $i$. Since each nice path (and thus, each relevant nice path) is sampled with probability $\Theta(2^i\log^3 n/\beta)$, with high probability the number of relevant nice paths intersecting the suffix $R(s_i,t)$ of $R(s,t)$ is at most $24\beta/( 2^i\log^2 n)$ (where the asymptotic notation is removed by adjusting the constant within the sampling probability) . 
\end{proof}

Because \cref{claim:numpaths} holds with high probability and it is easy to check that we only apply it a polynomial number of times, for the remainder of the proof we will assume that it holds deterministically.

By \cref{claim:numpaths}, the total additive error incurred during phase $i$ is at most \[\frac{3\eps\cdot 2^i}{\beta}\dist(s,t)\cdot \frac{24\beta}{ 2^i\log^2 n}=\frac{72\eps}{\log^2 n}\dist(s,t).\]

Taking the sum over all $\log n$ phases, the total additive error is 
\[\frac{72\eps}{\log n}\dist(s,t).\]

\paragraph[hi]{Distance Error from \hyperlink{s3}{(S3)}.} To get from $v$ to $s_{i+1}$, $Q_i$ first takes a shortest path from $v$ to $u$ (using $H(P)$), and then takes a path from $u$ to $s_{i+1}$ given by the vertex-path hopset procedure, which guarantees length at most $(1+\eps/2)\dist(u,s_{i+1})+\len(P)$. By the triangle inequality, $\dist(u,s_{i+1})\leq \dist(u,v_0)+\dist(v_0,v) + \dist(v,s_{i+1})$, so the total length of the subpath of $Q_i$ between $v$ and $s_{i+1}$ is at most \begin{align*}
 & \dist(v,u)+(1+\eps/2)(\dist(u,v_0)+\dist(v_0,v) + \dist(v,s_{i+1}))+\len(P)& \\
  \leq &(1+\eps/2)\dist(v,s_{i+1})+(2+\eps/2)\len(P) + (1+\eps/2)\dist(v_0,v) &\text{since $\dist(v,u)+\dist(u,v_0)\leq \len(P)$}\\
  \leq & (1+\eps/2)\dist(v,s_{i+1})+\eps  \dist(s,t)/(20\log n) + (3/2)\dist(v_0,v)&\text{by the definition of a relevant path.}\\
\leq & (1+\eps/2)\dist(v,s_{i+1})+\dist(s,t)\cdot(\eps /(20\log n) + \eps/(6 \log n)), &
\end{align*}
where the last inequality is because since we didn't take case (S1) while executing the loop on $v_0$ (and we executed the loop on $v_0$ by \cref{claim:first}).

That is, the segment of $Q_i$ from $v_0$ to $s_{i+1}$ has multiplicative error $(1+\eps/2)$ and additive error at most $\eps \dist(s,t)/(4\log n)$. 

This error occurs once per phase. Taking the sum over all $\log n$ phases, the additive error is at most $\eps\dist(s,t)/4$.% where the sum over the terms of the form $\dist(v_0,s_{i+1})$ became $\dist(s,t)$ since each $v_0$ to $s_{i+1}$ interval is a disjoint interval of $R(s,t)$. \nicole{todo simplify the expression earlier and make only in terms of $\eps$}.
\\

Finally, we take the the sum over all sources of additive error: 
\[\eps\dist(s,t)\cdot\big(1/(54\log n)+1/6+72/\log n+1/4\big)\]\[\leq \eps\dist(s,t)/2 \text{ \hspace{3mm}for sufficiently large $n$.}\]

\subsubsection{Hopbound Analysis}

Our goal is to show that $|Q|\leq \beta$. %First, we will consider the vertices on $Q\cap R(s,t)$, and later we will consider the rest of the vertices on $Q$. 
\paragraph{Hops on Shortcut Paths.}

First, we will count the vertices on $Q$ that are on the shortcut paths specified by cases \hyperlink{s1}{(S1)}-\hyperlink{s3}{(S3)}. Case \hyperlink{s3}{(S3)} incurs a total of at most $3\log n$ hops because it is executed only once per phase, each time incurring only 3 hops by the guarantee of the vertex-path hopset procedure. 

For cases \hyperlink{s1}{(S1)} and \hyperlink{s2}{(S2)}, for each phase $i$, we only enter these cases once per relevant nice path that intersects $R(s_i,t)$. By \cref{claim:numpaths}, the number of relevant nice paths intersecting $R(s_i,t)$ is at most $24\beta/(2^i\log^2 n)$. Taking the sum over all $\log n$ phases, we enter cases \hyperlink{s1}{(S1)} and \hyperlink{s2}{(S2)} a total of at most $24\beta/\log n$ times. Each time we invoke \hyperlink{s1}{(S1)} or \hyperlink{s2}{(S2)}, we either take a $\leq 2$-hop path from $H(P)$ or a $\leq 6$-hop from the backward shortcutting subroutine (\cref{lem:back}), for a total of $144\beta/\log n$ hops.

%It remains to count vertices $v$ on which we execute the loop, but do not take a shortcut path via cases \hyperlink{s1}{(S1)}-\hyperlink{s3}{(S3)}. First, we 

\paragraph{Hops Not on Relevant Paths.}

We will count the number of vertices on $Q\cap R(s,t)$ that are not on any relevant nice path.

% There are two types of vertices $v$ on which we execute the loop: those that are on some relevant nice path and those that are not. First, we will count those that are not:

% First, we calculate the number of vertices on $R(s,t)$ that are not on any relevant nice path:

\begin{claim}\label{claim:nopath}
The total number of vertices on $R(s,t)$ that are not on any relevant nice path is at most $\beta/3$.
\end{claim}

\begin{proof}
Suppose for contradiction that there is a set $S$ of more than $\beta/3$ vertices on $R(s,t)$ that are not contained in any relevant nice path. Consider the path $P_S$ in $G^*$ on the vertices of $S$ in the order that they appear on $R(s,t)$. $P_S$ is a shortest path by the definition of $R(s,t)$. Thus, any subpath of $P_S$ on exactly $\eps\beta/(162\log n)$ hops satisfies properties \ref{prop1}-\ref{prop3} of the nice path collection and is thus eligible to be added to the nice path collection.

We claim that there exists a subpath $P'$ of $P_S$ on $\eps\beta/(162\log n)$ hops of length at most $\eps\dist(s,t)/(54\log n)$. Indeed, if every subpath of $P_S$ on $\eps\beta/(162\log n)$ hops had length more than $\eps\dist(s,t)/(54\log n)$, then since $P_S$ has $\beta/3$ hops, $\len(P_S)$ would be more than \[\frac{\beta / 3}{\eps\beta/(162\log n)}\cdot \frac{\eps\dist(s,t)}{54\log n}=\dist(s,t),\] but we know that $\len(P_S)$ is in fact at most $\dist(s,t)$. %Let $P'$ be a subpath of $P_S$ of length at most $\eps\dist(s,t)/8$. %As stated previously, $P'$ is eligible to be added to the nice path collection $\mathcal{P}$. 

Let $j$ be the index such that the nice path $P_j$ is relevant while $P_{j+1}$ is not. (Note that there is only one such index $j$ since the nice path collection is ordered by length.) That is,  $\len(P_{j+1})> \eps \dist(s,t)/(54\log n)$. By property \ref{prop4} of the nice path collection, $\len(P_{j+1})$ is the shortest among all paths that are eligible to be added to the nice path collection $\{P_1,\dots,P_j\}$. However, this is a contradiction because $P'$ is eligible and has length less than that of $P_{j+1}$.
\end{proof}

\paragraph{Hops at Beginning and End.}

We will count the vertices on $Q_0$ and $Q_{\log n}$. $Q_{\log n}$ is defined as $R(s_{\log n},t)$.
%, and $s_{\log n}$ is defined as the last vertex on $R(s,t)$ belonging to a relevant nice path that is sampled at level $\log n$. 
By \cref{claim:numpaths}, the number of relevant nice paths intersecting $R(s_{\log n},t)$ is at most $24\beta/( n\log^2 n)<1$. Thus, all vertices on $R(s_{\log n},t)$ are not on any relevant nice path, so the hops on $Q_{\log n}$ are subsumed into the $\beta/3$ hops from \cref{claim:nopath}.

$Q_0$ is defined as $R(s,s')$ concatenated with the path from $s'$ to $s_1$ on at most 3 hops given by the vertex-path hopset procedure. The vertex $s'$ is defined as the first vertex on $R(s,t)$ that is sampled on level 1. %The number of vertices sampled on level 1 is $\theta(n\log^2 n/\beta)$.
Since vertices on level 1 are sampled with probability $\Theta(\log^2 n/\beta)$, with high probability $|R(s,s')|\leq \beta/\log n$ (for a sufficiently large constant in the asymptotic notation for the sample size).

\paragraph{Hops on ``Hard Vertices''.}

It remains to count vertices $v$ that fall on a relevant nice path, such that we execute the loop on $v$, but do not take a shortcut path via cases \hyperlink{s1}{(S1)}-\hyperlink{s3}{(S3)}. Call such vertices \emph{hard vertices}. 

Fix a phase $i$ between 1 and $\log n - 1$. Let $\mathcal{P}'=\{P'_1,\dots\,P'_k\}$ be the set of relevant nice paths that intersect $Q_i$ (in no particular order). For each $j$ from 1 to $k$, let $v_j$ be the last vertex on $Q_i$ such that $v_j\in P'_j$ and $v_j$ is a hard vertex (if such a vertex $v_j$ exists).
% Call such vertices \emph{hard vertices}. For a hard vertex $v$, if we execute the loop on $v$ during phase $i$, we call $v$ a \emph{phase-$i$ hard vertex}.
% Fix a phase $i$ and a phase-$i$ hard vertex $v$. 
%For each $i$, let $P_i$ be the relevant nice path containing $v_i$. 
Let $v_{j_0}$ be the first vertex on $Q_i$ that is on $P'_j$. (Note that $v_j$ is required to be hard, but $v_{j_0}$ is not.) 

By \cref{claim:first}, we executed the loop on $v_{j_0}$. If we took case \hyperlink{s1}{(S1)} while executing the loop on $v_{j_0}$, then we would jump straight to the last vertex on $R(s,t)$ that is on $P'_j$, and we would never again execute the loop on a vertex on $P'_j$ in this phase, a contradiction. (Note that even if $v_j$ is the vertex that we jumped to, we still would not execute the loop on $v_j$ since at the very beginning of the loop, we take a step forward on $R(s,t)$.) Thus, we took case \hyperlink{s2}{(S2)} while executing the loop on $v_{j_0}$, so we jumped to the last vertex $v_{easy}$ on $R(s,t)$ such that $[v_{j_0},v_{easy}]$ is an easy interval (recall \cref{def:easyhard} for easy/hard intervals). Thus, we know that $v_j$ falls after $v_{easy}$ on $R(s,t)$, and so $[v_{j_0},v_j]$ is a hard interval. Thus, \[h_{P'_j}(v_j,v_{j_0}) \geq 2^i\cdot |R(v_{j_0},v_j) \cap P'_j|.\]
Taking the sum over all $j$ on both sides:
\begin{equation}\label{eqn:sum}\sum_{j=1}^{k} h_{P'_j}(v_j,v_{j_0}) \geq 2^i\cdot \sum_{j=1}^k|R(v_{j_0},v_j) \cap P'_j|.\end{equation}

Call a vertex a \emph{sample candidate} if it falls between $v_j$ and $v_{j_0}$ (inclusive) on $P'_j$ for some $j$. For all $j$, none of the sample candidates on $P'_j$ are sampled on level $i+1$, because if one were then we would have taken a shortcut path in case \hyperlink{s3}{(S3)} and ended the phase. By definition, the number of sample candidates on $P'_j$ is $h_{P'_j}(v_j,v_{j_0})+1$. Since all of the nice paths are disjoint, the total number of sample candidates is at least $\sum_{j=1}^{k} h_{P'_j}(v_j,v_{j_0})$. 

Because none of the sample candidates are sampled on level $i+1$, and we sample vertices with probability $\Theta(\log^2 n/(2^i\beta))$ on level $i+1$, with high probability the total number of sample candidates is at most $2^{i-1}\beta/\log n$ (where the asymptotic notation is removed by adjusting the constant in the sampling probability). Since this bound on the number of sample candidates holds with high probability and we will only apply it a polynomial number of times, for the remainder of the argument we will assume that it holds deterministically. Thus, we have:
\[\sum_{j=1}^{k} h_{P'_j}(v_j,v_{j_0})\leq 2^{i-1}\beta/\log n.\]
Combining this with \cref{eqn:sum}:
\[\frac{\beta}{2\log n}\geq \sum_{j=1}^k|R(v_{j_0},v_j) \cap P'_j|.\]

Note that every hard vertex on $Q_i$ is contained in $|R(v_{j_0},v_j) \cap P'_j|$ for some $j$. This means that the number of hard vertices on $Q_i$ is at most $\beta/(2\log n)$. Taking the sum over all $\log n$ phases, the total number of hard vertices on $Q$ is at most $\beta/2$.
\\

Putting everything together, the total number of hops on $Q$ is:

\[3\log n+144\beta/\log n+\beta/3+\beta/\log n +3+\beta/2\]\[\leq\beta.\]

\subsection{Proof of Backward Shortcutting Subroutine}

The goal of this section is to prove \cref{lem:back}:
\lemback*

\subsubsection{Weaker Bound} Before proving \cref{lem:back}, we will first prove a claim that we will apply as a black box to prove \cref{lem:back}. The claim has a weaker approximation bound than \cref{lem:back} and concerns distances between pairs of vertices when at least one of the two vertices comes from a prespecified set $S$.

\begin{claim}\label{lem:weak}
Given any segment $P$ of a nice path, a subset $S\subseteq V(P)$ of vertices, and $\gamma,\delta\in (0,1)$, it is possible to add a set of  \[O\Big(\frac{|S|\cdot\log (nW)} {\gamma\delta}\Big)\] hopset edges to $G_{aug}$ so that for all vertices $x,y\in P$ such that $x$ appears after $y$ on $P$ and at least one of $x$ or $y$ is in $S$, the resulting graph has  an $xy$-path with 3 hops and length at most
\[(1+\gamma)\dist(x,y) + \delta\cdot\len(P).\]
\end{claim}

\begin{proof}We begin by defining the construction of hopset edges.
\paragraph{Construction.} Divide $P$ into $O(1/\delta)$ intervals $I_1,\dots,I_{O(1/\delta)}$ of (weighted) length at most $\delta\cdot\len(P)$ each. For each vertex $v\in S$, and each interval $I_i$, do the following. For all $j$ from 0 to $\log_{1+\gamma}(nW)$, add the following two hopset edges: \begin{itemize} \item the edge $(v,w)$ where $w$ is the first vertex on $I_i$ such that $\dist(v,w)\in[(1+\gamma)^j, (1+\gamma)^{j+1}]$ \item the edge $(w',v)$ where $w'$ is the last vertex on $I_i$ such that $\dist(w',v)\in [(1+\gamma)^j, (1+\gamma)^{j+1}]$.
\end{itemize}

\paragraph{Size of Hopset.} For each of $S$ vertices and each of $1/\delta$ intervals, we add $\log_{1+\gamma} \len(P)$ edges. Thus, the total number of edges is $|S|\cdot\log (nW) / (\gamma\cdot\delta)$.

\paragraph{Hopbound and Approximation Bound.} Fix vertices $x,y$. First, suppose $x \in S$. Let $i$ be such that $y$ falls into interval $I_i$. Let $j$ be such that $d(x,y) \in [(1+\gamma)^j, (1+\gamma)^{j+1}]$. Then we added an edge from $x$ to the first vertex $w$ on $I_i$ such that $d(x,w)  \in [(1+\gamma)^j, (1+\gamma)^{j+1}]$. By choice of $w$, we know that $w$ appears before $y$ on $P$. Therefore, there is a path $P'$ in $G_{aug}$ from $w$ to $y$ with 2 hops and distance exactly $\dist(w,y)\leq \delta\cdot \len(P)$ since $w$ and $y$ are in the same interval $I_i$. Consider the path formed by taking the edge from $x$ to $w$, and then taking the path $P'$ from $w$ to $y$. This path has 3 hops and distance at most \begin{align*}
    \dist(x,w)+\dist(w,y) &\leq (1+\gamma)^{j+1} + \delta\cdot \len(P)\\ &\leq  (1+\gamma)\dist(x,y) + \delta\cdot \len(P).
\end{align*}

The remaining case when $y \in S$ is symmetric to the case when $x \in S$ but we describe it explicitly for completeness. Let $i$ be such that $x$ falls into interval $I_i$. Let $j$ be such that $d(x,y) \in [(1+\gamma)^j, (1+\gamma)^{j+1}]$. Then we added an edge from the latest vertex $w'$ on $I_i$ such that $d(w',y)  \in [(1+\gamma)^j, (1+\gamma)^{j+1}]$. By choice of $w'$, we know that $w'$ appears after $x$ on $P$. Therefore, there is a path $P'$ from $x$ to $w'$ with 2 hops and distance exactly $\dist(x,w')$. Consider the path formed by taking the path $P'$ from $x$ to $w'$, followed by the edge from $w'$ to $y$. This path has 3 hops and distance at most $\dist(x,w')+\dist(w',y) <= \delta\cdot\len(P) + (1+\gamma)^{j+1} <=  (1+\gamma)\dist(x,y) + \delta\cdot\len(P)$.
\end{proof}

\subsubsection{Final Bound}

\paragraph{Construction Idea.} First, we will divide $P$ into intervals with an equal number of vertices per interval.  Then, we will add hopset edges to each interval individually. The intervals with larger than average (weighted) length require more hopset edges to achieve the desired distance approximation. We can indeed afford to add more hopset edges to the longer intervals because there are only few intervals substantially longer than the average length. Specifically, to add hopset edges to each interval, we will apply \cref{lem:weak} as a black box, setting $\delta$ to be smaller for longer intervals. Lastly, we will repeat the entire procedure for each of $\log n$ different values for the number of vertices per interval. 

\paragraph{Construction.} For each $k$ from 1 to $\log n$, split $P$ into $O(|P|/2^{k+1})$ intervals of $2^{k+1}$ vertices each, so that each interval is offset by $2^k$ from the previous interval. For a given $k$, let $\mathcal{I}_k$ be the collection of intervals.

\begin{observation}\label{obs:tog}
For any pair of vertices $u,v$ with $h_P(x,y) < 2^k$, there is at least one interval in $\mathcal{I}_k$ containing both $x$ and $y$, and at most two such intervals.
\end{observation}

Now, we will define some notation that will be useful. Let $\mu_k$ be the average (weighted) length of an interval in $\mathcal{I}_k$. Note that $\mu_k = \Theta(\len(P)\cdot 2^k/|P|)$. Given an interval $I\in \mathcal{I}_k$, define $\ell(I)$ to satisfy $\len(I) \in [2^{\ell(I)}, 2^{\ell(I) + 1}]\cdot\mu_k$.

% The idea is as follows. We will add hopset edges to each interval individually. The intervals with larger than average (weighted) length require more hopset edges to obtain the desired distance approximation. Because there are few intervals with substantially larger than average length, we can indeed afford to add more hopset edges to them. 
%In particular, for any fixed j, consider the intervals of weight [2^j, 2^{j+1}] *average. Only a 1/2^j fraction of intervals, i.e. |P|/2^{k+j} can have this weight. 
% To add hopset edges to each interval, we will apply \cref{lem:weak} as a black box, setting $\delta$ to be smaller for longer intervals.

To apply \cref{lem:weak}, we need to define the input parameters. We use $P^*$, $S^*$, $\gamma^*$, and $\delta^*$ to denote these input parameters (to distinguish them from our local values of $P$, $\gamma$, and $\delta$). For each $k$ we do the following. For each interval $I\in \mathcal{I}_k$, we apply \cref{lem:weak} $k$ times, each time with a different value for $S^*$. In particular, for each $i$ from $1$ to $k$, let $S_i$ be a set of vertices constructed by sampling each vertex in $I$ independently with probability $\log n/2^i$, and let $S^*=S_i$. Note that with high probability $|S_i|=\Theta(2^{k-i+1} \cdot \log n)$. The rest of the parameters are set as follows: $P^*=I$, $\gamma^*=\gamma$, $\delta^*=\delta/2^{\ell(I)+i+3}$.

\paragraph{Size of Hopset.} For an instantiation of \cref{lem:weak} with given $k$, $I$, $i$, The number of hopset edges added is:
\begin{align*}
    O\Big(\frac{|S_i|\cdot\log (nW) }{\gamma^*\delta^*}\Big) 
&= O\Big(\frac{2^{k-i+1}\cdot \log (nW) \log n }{\gamma \delta/2^{\ell(I)+i+3}}\Big)\\
&= O\Big(\frac{\log (nW)\log n \cdot 2^{k+\ell(I)}}{\gamma\delta}\Big).
\end{align*}

Taking the sum over all $i$, the total number of hopset edges added over all instantiations of \cref{lem:weak} for fixed $k$, $I$ is:
\[O\Big(\frac{\sum_{i=1}^k 2^{k+\ell(I)}\cdot\log(nW)\log n }{\gamma\delta}\Big)
=O\Big( \frac{2^{k+\ell(I)}\cdot\log(nW)\log^2 n }{\gamma\delta}\Big)\]

 Note that at most a $1/2^{\ell(I)}$ fraction of the intervals in $\mathcal{I}_k$ can have length in $[2^{\ell(I)}, 2^{\ell(I)+1}] \cdot\mu_k$, that is, only $O(|P|/2^{k+\ell(I)})$ intervals. Thus, the total number of hopset edges added over all intervals in $\mathcal{I}_k$ with weight in $[2^\ell(I), 2^{\ell(I)+1}] \cdot \mu_k$ is: 
\[O\Big(\frac{|P|}{2^{k+\ell(I)}}\cdot \frac{2^{k+\ell(I)}\cdot\log(nW)\log^2 n }{\gamma\delta}\Big)=O\Big(\frac{|P| \log(nW)\log^2 n }{\gamma\delta}\Big)\]

Taking the sum over all $O(\log(nW))$ possible values of $\ell(I)$, the total number of hopset edges for a given $k$ is:

\[O\Big(\frac{|P| \log^2(nW)\log^2 n }{\gamma\delta}\Big).\]

Taking the sum over all $k$, the total number of hopset edges is:

\[O\Big(\frac{|P| \log^2(nW)\log^3 n }{\gamma\delta}\Big).\]

%That is, $O(|P|/2^{k+j})$ intervals can have length in $[2^j, 2^{j+1}]\cdot \mu=O(\len(P)\cdot 2^{k+j}/|P|)$. The average interval length is $\mu = O(\len(P)\cdot 2^k/|P|)$.

\paragraph{Hopbound and Approximation Bound.} Fix vertices $x,y$ such that $x$ appears after $y$ on $P$. Our goal is to identify an $xy$-path that satisfies the guarantee of \cref{lem:back}.

The following claim will be useful.

\begin{claim}\label{claim:outside}
At most 4 vertices in $R(x,y) \cap P$ fall outside of the subpath of $P$ from $y$ to $x$.
\end{claim}

\begin{proof}
Consider a vertex $w$ in $R(x,y) \cap P$ that appears after $x$ on $P$. By definition, $G_{aug}$ has a 2-hop shortest path from $x$ to $w$. $R(x,y)$ is defined as a shortest path in $G_{aug}$ from $x$ to $y$ with the fewest hops among all shortest $xy$-paths in $G_{aug}$. Thus, the portion of $R(x,y)$ from $x$ to $w$ has at most 2 hops. That is, there are at most 2 vertices in $R(x,y) \cap P$ that appear after $x$. By a symmetric argument, there are at most 2 vertices in $R(x,y) \cap P$ that appear before $y$.
\end{proof}

Set $k$ and $i$ as follows:
\begin{itemize}
\item Let $k$ be such that $h_P(y,x) \in [2^{k-1}, 2^k]$. By \cref{obs:tog}, $x$ and $y$ appear together in some interval $I\in \mathcal{I}_k$. Fix $I$. (If there are two choices for $I$, pick one of them arbitrarily.)
%\item Let $j$ be such that $\len(I) \in [2^j, 2^{j+1}] \cdot\mu$. Note that $\mu = O(\len(P)\cdot2^k/|P|)$. 
\item Let $i$ be such that $|R(x,y)  \cap P| - 4 \in (2^i, 2^{i+1}]$. 
\end{itemize}

Let $(R(x,y) \cap P)_I$ denote the set of vertices in $R(x,y) \cap P$ that fall into the interval $I$. In the following claim we will argue that $S_i$ hits $(R(x,y) \cap P)_I$. 

\begin{claim}\label{claim:hit}With high probability $S_i$ hits $(R(x,y) \cap P)_I$. 
\end{claim} 
\begin{proof}The subpath of $P$ from $y$ to $x$ falls entirely inside of the interval $I$, so by \cref{claim:outside}, at most 4 vertices in $R(x,y) \cap P$ fall outside of the interval $I$. That is, $|(R(x,y) \cap P)_I| \geq |R(x,y) \cap P| - 4$. Thus, by the definition of $i$, $|(R(x,y) \cap P)_I| \in (2^i, 2^{i+1}]$. 

% Recall that $S_i$ is a set of vertices of size $2^{k-i+1} \cdot \log n$ randomly sampled from the interval $I$, which contains $2^{k+1}$ vertices. That is, each vertex in $I$ is sampled to $S_i$ with probability $\log n/2^i$. 

Recall that $S_i$ is a set of vertices constructed by sampling each vertex in $I$ with probability $\log n/2^i$. Since $|(R(x,y) \cap P)_I| \in (2^i, 2^{i+1}]$, with high probability $S_i$ hits $(R(x,y) \cap P)_I$. 
\end{proof}

Because \cref{claim:hit} holds with high probability and it is easy to check that we only apply it a polynomial number of times, for the remainder of the proof we will assume that it holds deterministically.

Let $s$ be a vertex in $S_i \cap (R(x,y) \cap P)_I$, which exists by \cref{claim:hit}. We will now define an $xy$-path, and calculate its length and number of hops. Our instantiation of \cref{lem:weak} with $S^*=S_i$ and $P^*=I$ finds a path from $x$ to $s$ and a path from $s$ to $y$ since $s\in S_i \cap (R(x,y) \cap P)_I$. Our $xy$-path will be the concatenation of these two paths. This path has 6 hops. It remains to calculate the length of this path. 

We note that because $s$ is on $R(x,y)$, $s$ is on a shortest path from $x$ to $y$. Thus, the additive error incurred by our path is simply the sum of the additive error incurred by each of the two subpaths (from $x$ to $s$, and from $s$ to $y$). The length of the subpath from $x$ to $s$ given by \cref{lem:weak} is
\begin{align*}
(1+\gamma^*)\dist(x,s) + \delta^*\cdot\len(P^*)
&\leq (1+\gamma)\dist(x,s) + \frac{\delta}{2^{\ell(I)+i+3}} \cdot 2^{\ell(I)+1} \cdot \frac{\len(P)\cdot 2^k}{|P|}\\
&= (1+\gamma)\dist(x,s) + \frac{\delta \cdot \len(P) \cdot 2^{k-i-2}}{|P|}\\
&\leq (1+\gamma)\dist(x,s) + \frac{\delta}{2}\cdot\frac{\len(P)\cdot h_P(y,x)}{|P|\cdot|R(x,y) \cap P|}\text{\hspace{4mm}by definition of $k$ and $i$}.
\end{align*}

The length of the subpath from $s$ to $y$ given by \cref{lem:weak} is exactly the same except with $\dist(x,s)$ replaced with $\dist(s,y)$.
%(1+\gamma)dist(s,y) + (\delta/2)*(w(P)*h(y,x)) / (|P|*|R(x,y) \cap P|).
Thus, the length of the concatenation of these two subpaths is 
\[(1+\gamma)\dist(x,y) + \delta\cdot\frac{\len(P)\cdot h_P(y,x)}{ |P|\cdot|R(x,y) \cap P|}\]
as desired.

\subsection{The $\beta>n^{1/3}$ Regime}
\label{sec:large-beta}

So far we have proved the first bound of Theorem \ref{thm:main}; that is, the bound for $\beta \leq n^{1/3}$. To prove the second bound ($\beta > n^{1/3})$, we use the first bound as a black box. In particular, we apply exactly the same argument (verbatim) as Section 2.2 of \cite{KP},  except instead of applying their shortcutting algorithm in Step 3, we apply our hopset algorithm on the graph $G'$ with parameter $\beta=(n')^{1/3}/\log n$, and passing in the original value of $\eps$. Their diameter bound analysis suffices to prove both our hopbound and distance approximation.

\section{Acknowledgements}
% DOLATER We would like to thank Shimon Kogan and Merav Parter for sharing with us their work in \cite{KP22focs} and \cite{KP22unpublished} and for numerous extremely helpful discussions.
We would like to thank Shimon Kogan and Merav Parter for sharing with us their work in \cite{KP22focs} and for numerous extremely helpful discussions. We would also like to thank Shyan Akmal for pointing out an issue in a proof from an earlier version of this paper.

\bibliographystyle{alpha}
\bibliography{bib.bib}

\appendix

\input{appendix.tex}

\end{document}

%% file: abstract.tex
For an n-vertex directed graph $G = (V,E)$, a $\beta$-\textit{shortcut set} $H$ is a set of additional edges $H \subseteq V \times V$ such that $G \cup H$ has the same transitive closure as $G$, and for every pair $u,v \in V$, there is a $uv$-path in $G \cup H$ with at most $\beta$ edges. A natural generalization of shortcut sets to distances is a $(\beta,\eps)$-\emph{hopset} $H \subseteq V \times V$, where the requirement is that $H$ and $G \cup H$ have the same shortest-path distances, and for every $u,v \in V$, there is a $(1+\eps)$-approximate shortest path in $G \cup H$ with at most $\beta$ edges.

There is a large literature on the question of the tradeoff between the optimal size of a shortcut set / hopset and the value of $\beta$. In particular we highlight the most natural point on this tradeoff: what is the minimum value of $\beta$, such that for any graph $G$, there exists a $\beta$-shortcut set $H$ with $O(n)$ edges? Similarly, what is the minimum value of $\beta$ such that there exists a $(\beta,\epsilon)$-hopset with $O(n)$ edges? Not only is this a very natural structural question in its own right, but shortcuts sets / hopsets form the core of a large number of distributed, parallel, and dynamic algorithms for reachability / shortest paths.

A lower bound of Hesse [SODA 2003] for directed graphs shows that if we restrict ourselves to hopsets with $O(n)$ edges, the best we can guarantee is $\beta = \Omega(n^{1/17})$ for both shortcut sets and hopsets [SODA 2003]; this was later improved to $\beta = \Omega(n^{1/6})$ by Huang and Pettie [SWAT 2018]. Until very recently the best known upper bound was a folklore construction showing $\beta = O(n^{1/2})$, but in a breakthrough result Kogan and Parter [SODA 2022] improve this to $\beta = \otil(n^{1/3})$ for shortcut sets and $\otil(n^{2/5})$ for hopsets. 

Our result in this paper is to close the gap between shortcut sets and hopsets introduced by the result of Kogan and Parter. That is, we show that for any graph $G$ and any fixed $\eps$ there is a $(\otil(n^{1/3}),\eps)$ hopset with $O(n)$ edges. Our hopset improves upon the $(\otil(n^{2/5}),\eps)$ hopset of Kogan and Parter.
More generally, we achieve a smooth tradeoff between hopset size and $\beta$ which exactly matches the tradeoff of Kogan and Parter for the simpler problem of shortcut sets (up to polylog factors). 

Additionally, using a very recent black-box reduction of Kogan and Parter, our new hopset immediately implies improved bounds for approximate distance preservers. 
%In particular, we show that for any weighted directed graph $G$ and a set of $p$ vertex pairs, there is a $(1+\eps)$-approximate distance preserver with 
%\aaron{Maybe add a sentence that this has implications for distance preservers.}

%% file: intro.tex
\section{Introduction}

Computing reachability and shortest paths in a directed graph is one of the most fundamental problems in graph algorithms. In a wide range of settings, these problems turn out to be easier if the path in question contains few edges (even if the edges have high weight). This dependency motivated the notion of \textit{shortcut sets}, introduced by Thorup \cite{Thorup92}. Given a graph $G$, the goal is to add a new set of edges $H$ such that $G \bigcup H$ has the same transitive closure as $G$, but for any pair of vertices $x,y$ there is an $xy$-path in $G \bigcup H$ with few edges. A generalization of a shortcut set is the notion of a \textit{hopset} $H$, where the goal is to preserve not just reachability but (weighted) distances.  We now formally define both these notions. 

%\begin{definition}[shortest distances]
%Given a graph $G$ with non-negative edge wegiths and two vertices $u,v$, we let $\dist_G(u,v)$ refer to the shortest distance from $u$ to $v$. 

%and $h_G(u,v)$ refer to the number of edges on the shortest path from $u$ to $v$. If there is no path from $u$ to $v$ we define $\dist(u,v) = h(u,v) = \infty$. Note that in an unweighted graph $\dist(u,v) = h(u,v)$. 
%\end{definition}

\begin{definition}[Shortcut Set]
\label{def:shortcut-set}
Given an unweighted graph $G=(V,E)$, a $\beta$-shortcut set is a set of edges $H \subseteq V \times V$ such that for all $u,v \in V$ the following holds: {\bf (1)} $u$ can reach $v$ in $G \cup H$ if and only if $u$ can reach $v$ in $G$ and {\bf (2)} if $u$ can reach $v$ in $G \cup H$ then there is a $uv$-path in $G \cup H$ with at most $\beta$ edges.
\end{definition}

%weighted edges such that for all vertices $u,v\in V$, \[\dist_G(u,v)\leq \dist_{G\cup H}(u,v)\leq\dist_{G\cup H}^{(\beta)}(u,v)\leq (1+\eps)\dist_G(u,v).\]

\begin{definition}[Hopset]
\label{def:hopset}
Given a graph $G=(V,E)$ with non-negative weights, a $(\beta,\eps)$-hopset is a set of edges $H \subseteq V \times V$ with non-negative weights such that such that for all $u,v \in V$ the following holds: {\bf (1)}  $\dist_G(u,v) = \dist_{G \cup H}(u,v)$, where $\dist(u,v)$ refers to the weighted distance between $u$ and $v$. And {\bf (2)} If $u$ can reach $v$ then there is a $uv$-path $P_{uv}$ in $G \cup H$ such that $P_{uv}$ contains at most $\beta$ edges and the weight of $P_{uv}$ is at most $(1+\eps) \dist(u,v)$
\end{definition}

%\begin{definition}[Hopset]
%Given a weighted graph $G=(V,E)$, a $(\beta,\eps)$-hopset is a set of weighted edges $H \in V \times V$ such that for all vertices $u,v\in V$, \[\dist_G(u,v)\leq \dist_{G\cup H}(u,v)\leq\dist_{G\cup H}^{(\beta)}(u,v)\leq (1+\eps)\dist_G(u,v).\]
%\end{definition}

\paragraph{Motivation}
There are two trivial extremes for hopset construction. Setting $H = \emptyset$ yields a $(n,0)$ hopset. On the other hand, if for every pair of vertices $u,v$ we add an edge to $H$ with $w(u,v) = \dist(u,v)$ then this yields a $(1,0)$ hopset with $n^2$ edges. The interesting question is thus to achieve a tradeoff between the number of edges and the parameter $\beta$. Perhaps the most natural setting to consider is as follows: if we restrict $H$ to contain $O(n)$ edges, what is the minimum $\beta$ we can guarantee?

The above is a natural question in extremal graph theory in its own right. But hopsets also play a crucial role in computing reachability and shortest path in directed graphs in a wide variety of models of computation such as distributed, parallel, and dynamic algorithms \cite{KleinS97,HenzingerKN14,HenzingerKN15,ForsterN18,LiuJS19,Fineman18,GutenbergW20,BernsteinGW20,CaoFR20,CaoFR21,KarczmarzS21}. In all of these models, most state-of-the-art algorithms for computing shortest paths start by first computing a hopset $H$ and then computing shortest paths in $G \cup H$, taking advantage of the fact that these shortest paths are guaranteed to contain at most $\beta$ edges. Note that for this second step to be efficient, it is crucial that $\beta$ is small and that $H$ contain relatively few edges. This brings us back to the original question of what kind of trade-offs are possible between these parameters. This question is further subdivided into three subproblems, in increasing order of generalization: shortcut sets capture reachability but not distances, $(\beta,\eps)$ hopsets capture $(1+\eps)$-approximate distances, and $(\beta,0)$-hopsets capture exact distances. 

\subsection{Previous work}
A folklore randomized construction for hopsets, attributed to Ullman and Yannakakis \cite{UllmanY91} and refined by Berman et al. \cite{BermanRR10}, is to randomly sample a set of vertices $S \subseteq V$, and then add an edge of weight $\dist(u,v)$ for all pairs $u,v \in S$ such that $u$ reaches $v$. This yields a $(\beta, 0)$ hopset with $\otil(n^2/\beta^2)$ edges.  In particular, there is a $(\sqrt{n},0)$-hopset with $\otil(n)$ edges. Surprisingly, this simple construction still achieves the best known trade-off for $(\beta, 0)$ hopsets. Existing work on the problem thus focuses on the simpler problems of shortcut sets and $(\beta, \eps)$-hopsets.

In \emph{undirected} graphs both of these problems admit sparse hopsets with small $\beta$. Shortcut sets are trivial in undirected graphs, as one can simply add a star to each connected component. For $(\beta,\eps)$ hopsets, a series of different constructions \cite{KleinS97,ShiS99,Cohen00} culminated in the papers of Huang and Pettie \cite{HuangP19} and Elkin and Neiman \cite{ElkinN19}, which showed a $(n^{o(1)}, \eps)$ hopset with $O(n)$ edges and a $(O(1),\eps)$ hopset with $n^{1+o(1)}$ edges. where $\eps$ is any fixed constant. Note that these are clearly optimal up to $n^{o(1)}$ factors, and moreover a lower bound of Abboud, Bodwin, and Pettie \cite{AbboudBP18} showed that a $n^{o(1)}$ factor is necessary. See also \cite{BenLevyP20,ShabatN21} for follow-up work that refined the $n^{o(1)}$ factor. 

The focus of our paper is on \textit{directed} graphs, where such hopset tradeoffs are provably impossible. Improving upon a previous lower bound of Hesse \cite{Hesse03}, Huang and Pettie \cite{HuangP21} showed that there exist graphs such that that any $O(n)$-size shortcut set cannot reduce $\beta$ to below $\Omega(n^{1/6}$); they also show that any $O(m)$-size shortcut set cannot reduce $\beta$ to below $\Omega(n^{1/11})$, and follow up work by Lu, Williams, Wein, and Xu improves this to $\Omega(n^{1/8})$ \cite{LuWWX22}. In particular, these lower bounds show a polynomial separation between hopsets in directed and undirected graphs. Moreover, all of these lower bounds apply even to the simpler problem of shortcut sets.

Until extremely recently, the best known upper bound in directed graphs was the folklore algorithm mentioned above. This left a large gap between the best-known upper and lower bounds: focusing on the standard case of a shortcut set $H$ with $O(n)$ edges, the best known upper bound achieved $\beta = \otil(n^{1/2})$, while the best known lower bound was $\beta = \Omega(n^{1/6})$. Very recently, in a major breakthrough, Kogan and Parter presented the first hopset to go beyond the folklore construction \cite{KP}. They presented a smooth trade-off between $\beta$ and the size of the hopset, but again focusing on the standard case of a set $H$ with $O(n)$ edges, they showed that every graph contains a $\otil(n^{1/3})$-shortcut set and an $(\otil(n^{2/5}),\eps)$ hopset for any fixed $\eps$. 

The result of Kogan and Parter shows the possibility of going beyond $\beta = \sqrt{n}$ for both shortcut sets and $(1+\eps)$-approximate hopsets, but it also leaves a polynomial gap between these two settings (i.e. $n^{1/3}$ vs. $n^{2/5})$. Our contribution in this paper is to close this gap.

\subsection{Our Contribution}

\begin{theorem}\label{thm:main}
For any directed graph with integer edge weights in $[1,W]$, given $\eps \in (0,1)$ and $\beta\geq 20\log n$, there is a $(\beta,\eps)$-hopset $H$ of size
\[|H|=
\begin{dcases}
O\Big(\frac{n^2 \cdot\log^7 n \log^2(nW)}{\eps^2 \beta^3}\Big) & \text{for } \beta\leq n^{1/3},\\[1em]
O\Big(\frac{n^{3/2} \cdot\log^7n \log^2(nW)}{\eps^2 \beta^{3/2}}\Big) & \text{for } \beta> n^{1/3}.
\end{dcases}\]

\end{theorem}

Our result improves polynomially upon the previous-best tradeoff for $(\beta,\eps)$ hopsets of Kogan and Parter (see Theorem 1.5 in \cite{KP}). Most notably, for a hopset $H$ with $O(n)$ edges, and any fixed $\eps$, Kogan and Parter \cite{KP} construct a $(\otil(n^{2/5}), \eps)$ hopset, whereas we construct a $(\otil(n^{1/3}),\eps)$-hopset. Moreover, up to logarithmic factors, our tradeoff exactly matches that of Kogan and Parter for the simpler problem of $\beta$-shortcut sets. We thus close the gap between $(1+\eps)$-approximate hopsets and shortcut sets.

%As mentioned above, one of there are two main motivations for studying hopsets. The first is that many shortest path algorithms in a variety of models (e.g. distributed, parallel, dynamic) rely on hopsets, and the recent development of improved hopsets by Kogan and Parter and now in this paper has the potential to lead to faster algorithms in those settings. 

\paragraph{Construction Time}
The focus of this paper is on existential claims about hopset trade-offs, so we make no attempt to optimize the time required to construct the hopset. Nonetheless it is easy to check that the construction in this paper can be executed in polynomial time. On the other hand, a clear bottleneck in our construction is that it requires computing the transitive closure, so even if the algorithm is refined, our approach necessarily requires a runtime of $\Omega(mn)$. 
% DOLATER The hopset of Kogan and Parter \cite{KP} also requires $\Omega(mn)$ construction time, but they were later able to achieve significantly faster construction time of $\tilde{O}(mn^{1/3})$ for their \emph{shortcut} set \cite{KP22icalp,KP22unpublished}. 
The hopset of Kogan and Parter \cite{KP} also requires $\Omega(mn)$ construction time, but they were later able to achieve significantly faster construction time of $\tilde{O}(mn^{1/3} + n^{1.5})$ for their linear-size $\otil(n^{1/3})$- \emph{shortcut} set \cite{KP22icalp}. 
Achieving similarly fast construction times for the more general problem of hopsets remains an intriguing open problem.

\paragraph{Application to Distance Preservers}

In addition to their critical role in many shortest path algorithms (see above), one of the main motivations for studying hopsets is their close connection to other fundamental problems in extremal graph theory, such as spanners, emulators, and distance preservers. At first glance these problems appear somewhat different because hopsets are trying to \textit{augment} the graph with new edges, while spanners and distance preservers are trying to \textit{sparsify} the graph by removing edges; but in very intriguing recent work, Kogan and Parter prove a strong connection between these problems by showing a black-box conversion from construction of shortcut sets / hopsets to construction of reachability/distance preservers (Theorem 1.1 in \cite{KP22focs}). 

In particular, using this new black box from \cite{KP22focs}, the $(\beta,\eps)$-hopset of Kogan and Parter from \cite{KP} implies a $(1+\eps)$-approximate distance preserver with $\otil(np^{2/5} + n^{2/3}p^{2/3})$ edges, where $p = |P|$ is the number of pairs. Applying the same black box from \cite{KP22focs}, our improved $(\beta,\eps)$-hopset immediately implies a $(1+\eps)$-approximate distance preserver with $\otil(np^{1/3} + n^{2/3}p^{2/3})$ edges. For $p \geq n$, our bound matches (up to log factors) the state-of-the-art sparsity of $O(n + n^{2/3}p^{2/3})$ of Abboud and Bodwin for the simpler problem of reachability preservers \cite{AbboudBP18}.

For more discussion of the definition of distance preservers, and of how our result fits into existing work, see Section \ref{sec:app_preservers} in the appendix. 

\paragraph{Techniques}
The main challenge of extending the linear-size $\otil(n^{1/3})$-shortcut set of Parter and Kogan to a $(\otil(n^{1/3},\eps))$-hopset is that their framework crucially relies on the graph $G$ being a DAG. For shortcut sets this can be assumed without loss of generality, since one can easily reduce the case of a general graph $G$ to that of a DAG: compute a shortcut set on the graph $G$ with all strongly connected components contracted (this is a DAG), and add to the shortcut set a bidirectional star for each strongly connected component ($2n$ edges total). 

For hopsets, however, this reduction does not hold. If the input graph $G$ happens to be a DAG, then the ideas of Kogan and Parter easily extend to a $(\otil(n^{1/3},\eps))$-hopset; but for general graphs, we must augment their basic framework with an entirely new set of techniques.

See Section \ref{sec_overview} for a more detailed overview of our techniques.

\subsection{Notation}
Let $G=(V,E)$ be a directed graph with positive weights. For two vertices $u,v\in V$, let $\dist_G(u,v)$ denote the distance from $u$ to $v$ in $G$. Let $G^*$ be the weighted transitive closure of $G$ i.e. each edge $(u,v)$ in $G^*$ has weight $\dist_G(u,v)$. 

Given a shortest path $P\subseteq G$, let $\len(P)$ denote the (weighted) length of $P$. Let $|P|$ be shorthand for $|V(P)|$. For a pair of vertices $u,v\in P$ where $u$ appears before $v$ on $P$, let $h_P(u,v)$ be the number of hops on $P$ from $u$ to $v$.

%% file: appendix.tex
\section{Implications of our Improved Hopset for Distance Preservers}
\label{sec:app_preservers}

As discussed in the introduction, very recent work by Parter and Kogan presents a black-box transformation from hopsets to distance preservers \cite{KP22focs}. As a result, our new $(\beta,\eps)$-hopset leads to improved bounds for approximate distance preservers. We now discuss in more detail the relevant definitions and related work.

%In particular, their hopset from \cite{KP} immeidattely implied improved bounds for $(1+\eps)$-approximate distance preservers, and out improved hopset leads to a further improvement for $(1+\eps)$-approximate distance preservers. 

\paragraph{Definitions and Existing Results for Distance Preservers:}
Given a graph $G = (V,E)$ and a set of demand pairs $P \subset V \times V$, an \emph{exact distance preserver} is a sparse subgraph $\gstar \subseteq G$ such that $\dist_{\gstar}(u,v) = \dist_G(u,v)$ for all $u,v \in P$. There are two common relaxations: a $(1+\eps)$-\emph{approximate} distance preserver guarantees that for all pairs $u,v \in P$, $\dist_{\gstar}(u,v) \leq  (1+\eps)\cdot \dist_G(u,v)$; a \textit{reachability} preserver guarantees that for all pairs $u,v \in P$, if $u$ can reach $v$ in $G$, then $u$ can reach $v$ in $\gstar$.

Focusing on directed graphs, and letting $p = |P|$ be the number of vertex pairs, the state-of-the-art for exact distance preservers is a construction by Coopersmith and Elkin with $O(np^{1/2})$ edges \cite{CoppersmithE06}, and one by Bodwin with $O(n + n^{2/3} \cdot p)$ edges \cite{Bodwin21}. For the simpler problem of reachability preservers, Abboud and Bodiwn showed a distance preserver with $O(n + n^{2/3}p^{2/3})$ edges \cite{AbboudB18}.

%near-exact and reachability distance preservers are less well understood. Here we focus in particular on directed graphs. For exact distance preservers the best known upper bounds are a distance preserver $\gstar$ with $np^{1/2}$ edges and one with $n + np^{2/3}$ edges. For reachability preservers, the best upper bound is $O(n + n^{2/3}p^{2/3}$ edges.

Until extremely recently, there were no upper bounds for $(1+\eps)$-approximate distance preservers that improved upon the bounds for exact distance preservers. But in extremely recent (unpublished) work, Kogan and Parter gave a black-box transformation from hopsets to distance preservers (see Theorem 1.1 in \cite{KP22focs}), and in particular showed that their $(\beta,\eps)$-hopset from \cite{KP} implies a $(1+\eps)$-approximate distance preserver with $\otil(np^{2/5} + n^{2/3}p^{2/3})$ edges. For $p \geq n^{5/4}$, this improves upon the bounds for exact distance preservers, and in fact matches (up to log factors) the best known bounds for the simpler problem of reachability preservers. 

Applying the same black box of Kogan and Parter \cite{KP22focs}, our new $(\beta, \eps)$-hopset directly leads to a further improvement.

\begin{theorem}
For any directed graph $G = (V,E)$ with integer weights that are polynomial in $n$, and any fixed $\eps > 0$, one can compute a $(1+\eps)$-approximate distance preserver with $\otil(np^{1/3} + n^{2/3}p^{2/3})$ edges.
\end{theorem}

Our bound improves upon $\otil(np^{2/5} + n^{2/3}p^{2/3})$ of Kogan and Parter \cite{KP22focs}. Moreover, for all $p \geq n$ our bound matches (up to log factors) the sparsity of Abboud and Bodwin \cite{AbboudB18} for the simpler problem of reachability preservers.

\section{Path Hopset Impossibility}\label{app}

Here, we prove the following claim, which was referenced as motivation in the technical overview. The proof follows from the shortcut set lower bounds of Hesse \cite{Hesse03} (or Huang and Pettie \cite{HuangP21}) using a simple argument, and we include it here for completeness.

\begin{claim}Given a positive integer $n$ and a number $\eps\in [0,1]$, for any positive integer $\ell$ that is polynomial in $n$, there exists a weighted directed graph $G$ on $n$ vertices containing a path $P$ with $\ell$ hops such that it is impossible to add $\tilde{O}(|P|)$ edges $H$ to $G$ so that every pair of vertices $x,y\in P$ where $x$ reaches $y$ in $G$, has a subpolynomial-hop path of length at most $(1+\eps)\dist(x,y)$ in $G\cup H$.\end{claim}

\begin{proof}
Hesse \cite{Hesse03} proved that there exists a DAG with no shortcut set of size $\tilde{O}(n)$ with subpolynomial hopbound. Construct the graph $G$ as follows. Start with a copy of Hesse's DAG $D$ on $\ell$ vertices, where all edges have weight one; now, add additional edges to form a path $P$ such that $P$ has $\ell$ hops, the weight of each edge of $P$ is $2n$, and the vertices of $P$ are ordered in the reverse of the topological order of $D$. 

Suppose $G$ has a hopset $H$ as specified in the statement of the claim. Then, we claim that $D$ has a shortcut set of size $\tilde{O}(\ell)$ with subpolynomial hopbound (a contradiction). Fix a pair of vertices $s,t$ such that $s$ reaches $t$ in the DAG $D$. Consider the subpolynomial-hop path $Q$ in $G\cup H$ from $s$ to $t$ of length at most $(1+\eps)\dist_G(s,t)$. Since $D$ is unweighted, $\dist_G(s,t)\leq n$, so $(1+\eps)\dist_G(s,t)<2n$. Thus, $Q$ does not include any edges from $P$, nor any any edges in $H$ of weight $\geq 2n$. Every edge $u,v$ in $H$ has weight at least $\dist(u,v)$ because otherwise $\dist_{G\cup H}(u,v)<\dist_G(u,v)$, which violates the definition of a hopset. Thus, for every edge $(u,v)$ of $H$ with weight $<2n$, $u$ reaches $v$ in the DAG $D$. Therefore, $H$ is a valid shortcut set for $D$. 
\end{proof}

% it would even be sufficient to add $\tilde{O}(|P|)$ hopset edges so that every pair of vertices $x,y\in P$ has a small-hop path of length $(1+\eps')\dist(x,y)$ for sufficiently small $\eps'$.

% Hesse \cite{Hesse03} proved that there exists a DAG with no shortcut set (let alone a hopset) of size $\tilde{O}(n)$ with subpolynomial hopbound

% (This implies impossibility of our goal of obtaining a hopset for each nice path $P$ by the following construction: suppose each edge in $P$ has large weight, and overlay Hesse's DAG onto the vertices of $P$ where the topological order of the DAG is the reverse order of $P$.)